\documentclass{article}

\usepackage[shortlabels]{enumitem}
\usepackage{amssymb}
\usepackage{amsmath}
\usepackage{amsthm}
\usepackage{amsopn}
\usepackage{graphicx}
\usepackage{bbm}
\usepackage{mathdots}
\usepackage[export]{adjustbox}
\usepackage{mathdots}
\usepackage{hyperref}
\usepackage{float}
\usepackage{stmaryrd}

\hypersetup{
    colorlinks = true,
    citecolor=black,
    filecolor=black,
    linkcolor=black,
    urlcolor=black
    linktoc=all,     
    linkcolor=blue,  
}

\graphicspath{{images/}}

\newtheorem{theorem}{Theorem}[section]
\newtheorem{definition}[theorem]{Definition}
\newtheorem{claim}[theorem]{Claim}
\newtheorem{lemma}[theorem]{Lemma}
\newtheorem{conclusion}[theorem]{Conclusion}
\newtheorem{remark}[theorem]{Remark}

\newtheorem{example}[theorem]{Example}
\newtheorem*{conjecture}{Conjecture}
\newtheorem*{questions}{Questions}

\newcommand{\CC}{\mathbb C}

\newcommand{\RR}{\mathbb R}

\newcommand{\NN}{\mathbb N}
\newcommand{\ZZ}{\mathbb Z}

\newcommand{\mca}{\mathcal A}
\newcommand{\mcb}{\mathcal B}

\newcommand{\mcm}{\mathcal M}
\newcommand{\mce}{\mathcal E}
\newcommand{\mch}{\mathcal H}
\newcommand{\mci}{\mathcal I}
\newcommand{\mcj}{\mathcal J}

\newcommand{\mco}{\mathcal O}

\newcommand{\HH}{\mathcal H}

\newcommand{\TT}{\mathbb T}

\newcommand{\mcl}{\mathcal L}
\newcommand{\rmg}{\mathrm G}

\DeclareMathOperator{\rk}{rank}
\DeclareMathOperator{\ev}{ev}

\DeclareMathOperator{\ess}{ess}

\DeclareMathOperator{\trace}{trace}

\DeclareMathOperator{\End}{End}
\DeclareMathOperator{\diag}{diag}
\DeclareMathOperator{\Sp}{Span}

\DeclareMathOperator{\op}{op}
\DeclareMathOperator{\Id}{Id}

\title{On polynomials in spectral projections of spin operators}
\date{\vspace{-8ex}}
\begin{document}
\maketitle
\begin{center}{\author{Ood Shabtai\footnote{Partially supported by the European Research Council Starting grant 757585 and by the Israel Science Foundation grant 1102/20}}}\end{center}
\begin{abstract} We show that the operator norm of an arbitrary bivariate polynomial, evaluated on certain spectral projections of spin operators, converges to the maximal value in the semiclassical limit.
We contrast this limiting behavior with that of the polynomial when evaluated on random pairs of projections. The discrepancy is a consequence of a type of Slepian spectral concentration phenomenon, which we prove in some cases. \end{abstract}
\tableofcontents
\section{Introduction}
Let $J_1, J_2, J_3: \CC^n \to \CC^n$ be the generators of a unitary irreducible representation of $SU(2)$, satisfying the commutation relations
\begin{equation*} [J_1, J_2] = i J_3,\ [J_2, J_3] = i J_1,\ [J_3, J_1] = i J_2.\end{equation*}
Their spectrum equals the set $\sigma_n = \{j, j-1, ..., -j\}$, where $2j+1 = n$ (\cite{biedenharnlouck}).

Let $\mca$ denote the complex algebra generated by two non-commuting variables $x,y$ which satisfy $x^2 = x,\ y^2 = y$. Let $\CC[z,w]$ denote the commutative algebra of complex polynomials, and let
\begin{equation}\label{abelianizationmap} T : \mca \to \CC[z,w]\left/ \mci\right.\end{equation} map $f(x,y)$ to $f(z,w)$, where $\mci$ is the ideal generated by $z^2 - z,\ w^2 - w$.

Let $\left(\rmg_k(n), \mu_{k,n} \right)$ be the Grassmannian of $k$-dimensional subspaces of $\CC^n$, equipped with the uniform probability measure\footnote{i.e., the unique probability measure invariant under the action of the unitary group on $\rmg_k(n)$.}. We identify $V \in \rmg_k(n)$ with the orthogonal projection $P : \CC^n \to V$.

In this paper, we will study the following questions.
\begin{questions} Fix $f \in \mca$. Let $\mathbbm 1_{(0,\infty)}$ be the indicator function of $(0,\infty) \subset \RR$.
\begin{enumerate}
\item\label{q1}{What is the behavior of $f\left(\mathbbm 1_{(0,\infty)}(J_1), \mathbbm 1_{(0,\infty)}(J_3) \right)$ in the semiclassical limit $n \to \infty$?}
\item\label{q2}{How does it compare with $f(P,Q)$, where $P,Q \in \rmg_{\lfloor \frac n 2 \rfloor}(n)$ are random projections?}
\item{To what extent do the answers to questions \ref{q1}, \ref{q2} change if we replace $\mathbbm 1_{(0,\infty)}(J_1), \mathbbm 1_{(0,\infty)}(J_3)\in \rmg_{\lfloor \frac n 2 \rfloor}(n)$ with spectral projections of ranks $\lfloor \alpha n \rfloor$, where $0 < \alpha \le \frac 1 2$?}
\end{enumerate}\end{questions}
The non-commutative polynomials $f \in \mca$ provide an auxiliary tool to probe the limiting behavior of the relevant sequences of pairs of projections. In this context, certain polynomials are not very useful (e.g., constant polynomials). For this reason (and for the sake of conciseness), it will occasionally be fruitful to restrict our attention to $f \in \ker(T)$.

One of the main points of the paper is that pairs of spectral projections of spin operators are quite different from random pairs of projections (as $n \to \infty$). This will be demonstrated rigorously (Conclusion \ref{discrepancy}, Theorem \ref{prolate_theorem}), as well as through numerical simulations. The non-generic behavior is a consequence of a type of Slepian spectral concentration phenomenon, closely related to that of the prolate matrix and its variants (\cite{llwang, EMT, slepian, varah}). Namely, our simulations suggest that the principal angles between the ranges of the projections cluster near $0$ and $\frac \pi 2$, and we manage to prove this in some cases.

%

The present paper continues an earlier one (\cite{shabtai}), which also examined pairs of spectral projections arising from certain specific non-commuting quantum observables, and from spin operators in particular. However, previously we restricted our attention only to commutators (i.e., $f(x,y) = xy - yx)$, and did not compare with random projections. 

The analogues of our current, more general result on spin operators (namely, Theorem \ref{spin_theorem}) also hold for the rest of the cases considered in the previous paper (position and momentum operators, for instance). Ultimately, we suspect that these results are instances of a rather general phenomenon. We refer the reader to Section \ref{discussion_section} for further details.
\subsection{Main results}
We consider the case $\alpha = \frac 1 2$ first.  Denote $P_{1,n} = \mathbbm 1_{(0,\infty)}(J_1)$, $P_{3,n} = \mathbbm 1_{(0,\infty)}(J_3)$.
Let $\Omega_n = \rmg_{\left \lfloor \frac n 2  \right \rfloor} (n) \times \rmg_{\left \lfloor \frac n 2 \right \rfloor}(n)$, $\nu_n = \mu_{\left \lfloor \frac n 2 \right \rfloor, n} \times \mu_{\left \lfloor \frac n 2 \right \rfloor, n}$, and
\begin{equation*} I_n^f = \int_{\Omega_n} \left \Vert f\left(P, Q \right) \right \Vert_{\op} d\nu_n(P,Q).\end{equation*}
The following result was essentially conjectured by D. Kazhdan.
\begin{theorem}\label{theorem1} Let $0 \ne f \in \mca$. There exists a constant $M_f > 0$, depending only on $f$, such that
\begin{equation*}  \max_{(P,Q) \in \Omega_n} \left \Vert f(P, Q) \right \Vert_{\op} = \lim_{n \to \infty} I^f_n = \lim_{n \to \infty} \left \Vert f\left(P_{3,n}, P_{1,n} \right) \right \Vert_{\op} = M_f.\end{equation*}
$M_f$ is the universal tight upper bound (\ref{universal_upper_bound}) for $\Vert f(P,Q)\Vert_{\op}$, where $P,Q$ are arbitrary orthogonal projections on a separable Hilbert space.\end{theorem}
Theorem \ref{theorem1} is illustrated for $f(x,y) = xy - yx$ (where $M_f = \frac 1 2$) in the images below. The pair $P_{1,n}$, $P_{3,n}$ is non-generic as $n \to \infty$ (as follows from Theorem \ref{prolate_theorem}). Accordingly, the graphs in figures \ref{fig1} and \ref{fig2} are clearly dissimilar, even though the limits coincide. The seemingly different convergence rates are discussed in Remark \ref{conv_rate_conc}.
%
%
\begin{figure}[H]
    \centering
        \includegraphics[width=1\textwidth]{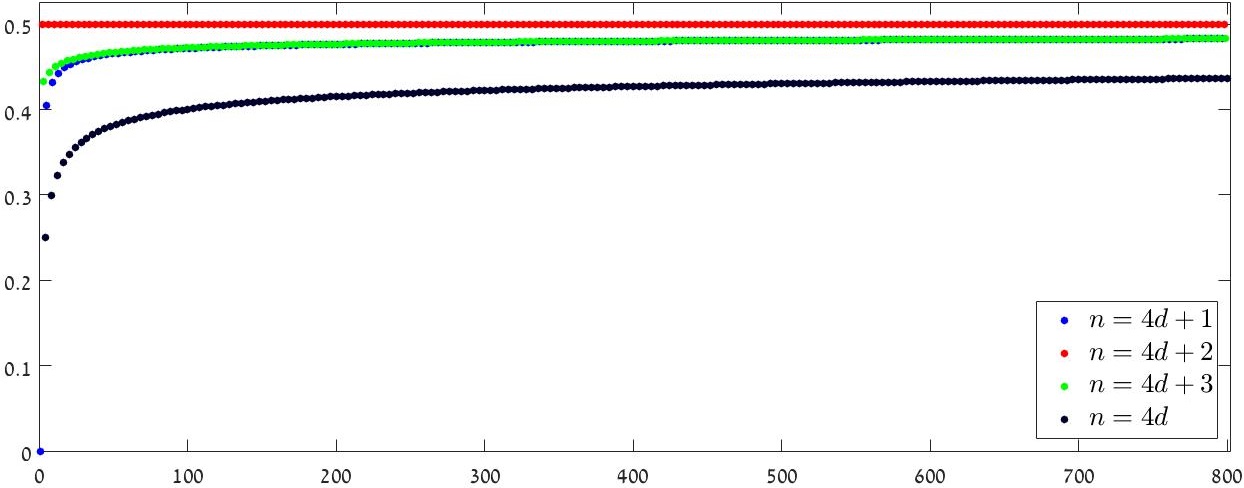}
        \caption{(Y. Le Floch) $\left \Vert \left[P_{1,n}, P_{3,n} \right] \right \Vert_{\op}$ as a function of $n$. The apparent mod $4$ behavior is unproven, except for the case $n = 4d+2$ (L. Polterovich).}\label{fig1}
\end{figure}
\begin{figure}[H]
    \centering
        \includegraphics[width=1\textwidth]{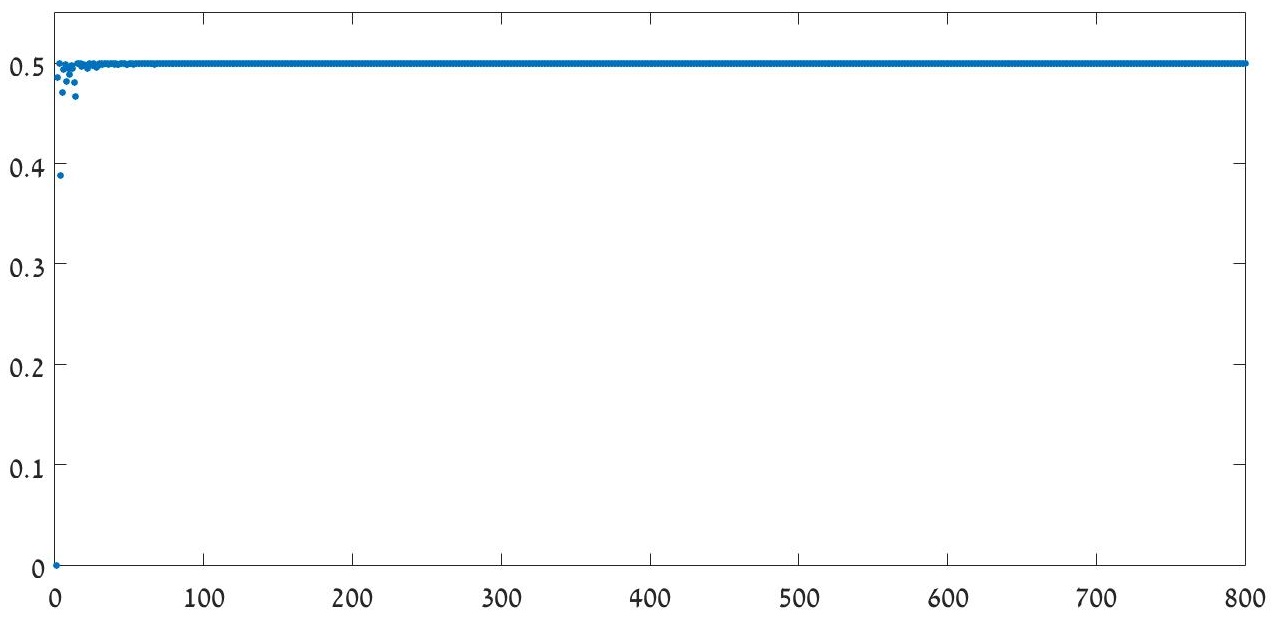}
        \caption{$\left \Vert \left[P, Q \right] \right \Vert_{\op}$ for random projections of rank $\lfloor \frac n 2 \rfloor$, as a function of $n$.}\label{fig2}
\end{figure}
Theorem \ref{theorem1} is in fact atypical, since the analogous limits do not necessarily coincide when the projections involved are of ranks $\lfloor \alpha n \rfloor$, with $\alpha < \frac 1 2$. Indeed, for $0 < \alpha \le \frac 1 2$, let $\Omega_{n,\alpha} = \rmg_{\lfloor \alpha n \rfloor}(n) \times \rmg_{\lfloor \alpha n \rfloor}(n)$, equipped with the product probability measure $\nu_{n,\alpha} = \mu_{\lfloor \alpha n \rfloor, n} \times \mu_{\lfloor \alpha n \rfloor, n}$. Denote
\begin{equation*} I_{n,\alpha}^f = \int_{\Omega_{n,\alpha}} \left \Vert f\left(P, Q \right) \right \Vert_{\op} d\nu_{n,\alpha}(P,Q).\end{equation*}
Then a combination of results from \cite{collins}, \cite{bottspitk} readily implies the following.
\begin{theorem}\label{grassmann_theorem} Let $f \in \ker (T)$ and $\lambda_\alpha = 4\alpha(1-\alpha)$. There exists a continuous $\psi_f : [0,1] \to [0,\infty)$, determined by $f$ only, such that $\psi_f(0) = \psi_f(1) = 0$, and
\begin{equation}\label{grassmann_limit}I^f_{\alpha} = \lim_{n \to \infty} I^f_{n,\alpha} = \max_{[0,\lambda_\alpha]} \psi_f.\end{equation}
Note that $M_f = \max_{[0,1]} \psi_f>0$ (for $f \ne 0$).
\end{theorem}
$\psi_f$ essentially appears in \cite{spitkovsky}. Since it is continuous, so is the non-decreasing function $\alpha \mapsto I^f_\alpha$.
\begin{conclusion}\label{small_grassmann} If $f \in \ker(T)$, then $\lim_{\alpha \to 0^+} I_\alpha^f = \psi_f(0) = 0$.\end{conclusion}
On the other hand,
\begin{theorem}\label{spin_theorem} Let $f \in \mca$. Define intervals $(0, \alpha_n)\subset \RR$ containing exactly $\left \lfloor \alpha n \right \rfloor$ elements of $\sigma_n$. Denote $P_{1,\alpha,n} =  \mathbbm 1_{(0, \alpha_n)}(J_1)$, $P_{3,\alpha, n} =  \mathbbm 1_{(0, \alpha_n)}(J_3)$. Then
\begin{equation}\label{spin_limit} \lim_{n \to \infty} \left \Vert f \left(P_{3,\alpha,n}, P_{1,\alpha,n}\right) \right \Vert_{\op} = M_f.\end{equation}\end{theorem}
In particular, when $\alpha < \frac 1 2$, we can rigorously say that pairs of spectral projections of spin operators are unlike random pairs of projections.
\begin{conclusion}\label{discrepancy} For every $\alpha < \frac 1 2$ there exists $f \in \ker(T)$ such that $I_\alpha^f < M_f$. For every $0 \ne f \in \ker(T)$ there exists $\delta > 0$ such that $I_\alpha^f < M_f$ for every $\alpha < \delta$. Thus, there is a discrepancy between (\ref{grassmann_limit}), (\ref{spin_limit}).\end{conclusion}
Conclusion \ref{discrepancy} is illustrated for $f(x,y) = xy - yx$, $\alpha = \frac 1 {20}$ in the following images. Note that $\psi_f(t) = \sqrt{t(1-t)}$, $M_f = \frac 1 2$ in this case (see Example \ref{example_psi_f}).
\begin{figure}[H]
    \centering
        \includegraphics[width=1\textwidth]{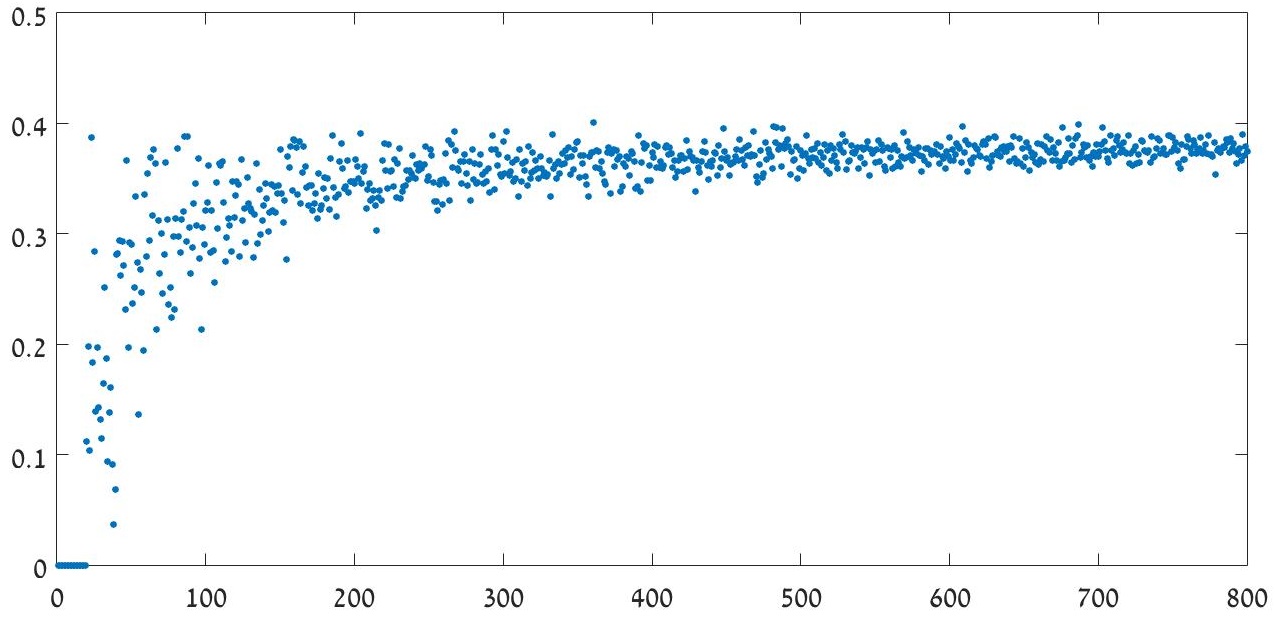}
        \caption{$\left\Vert \left[P_{\lfloor \alpha n \rfloor}, Q_{\lfloor \alpha n \rfloor} \right] \right\Vert_{\op}$ as a function of $n$ for random projections. Here, $\alpha = \frac 1 {20}$. The values concentrate about $2(1-2\alpha)\sqrt{\alpha(1-\alpha)} \approx 0.3923$.}\label{fig3}
\end{figure}
\begin{figure}[H]
    \centering
        \includegraphics[width=1\textwidth]{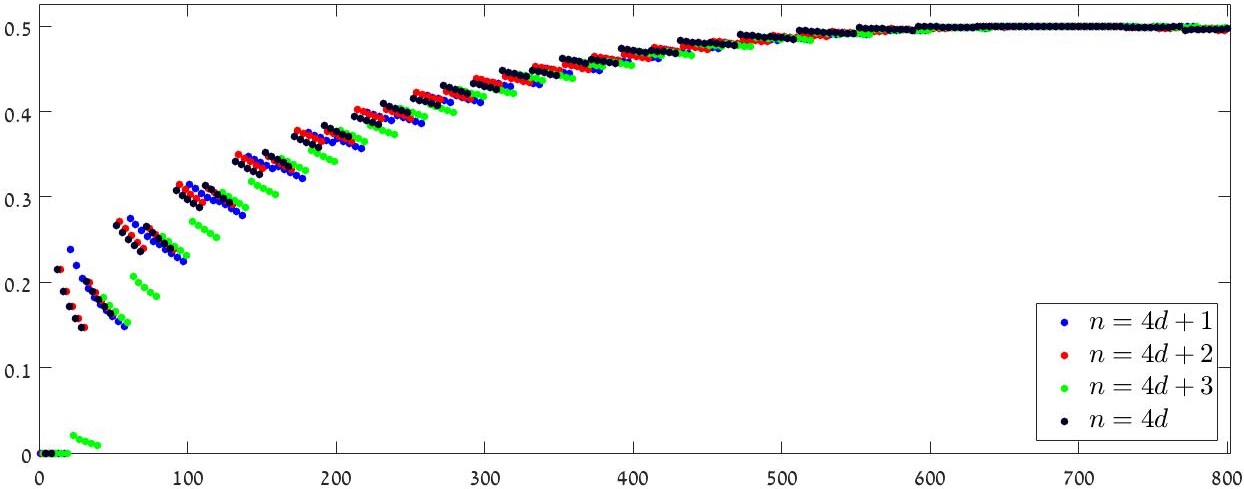}
        \caption{$\Vert \left[P_{1, \alpha, n}, P_{3,\alpha,n} \right] \Vert_{\op}$ as a function of $n$ for $\alpha = \frac 1 {20}$.}\label{fig4}
\end{figure}
\begin{remark} The last figure (\ref{fig4}) provides an example of a curious, unproven phenomenon. When $\alpha = \frac 1 {4k}$, our numerical simulations suggest that the graph of $n \mapsto \left \Vert \left[P_{1,\alpha,n}, P_{3,\alpha,n} \right] \right \Vert_{\op}$ decomposes, modulo 4, to pieces of length $k$. \end{remark}

Finally, we consider the pairs $P_{3,\alpha,n}$, $P_{1,n}$. Let $N_{n}(s,t)$ denote the number of eigenvalues of $P_{3,\alpha,n} P_{1,n} P_{3,\alpha,n} \in \End \left(\text{Im}(P_{3,\alpha,n}) \right)$ lying in the interval $[s,t]$, where $0 \le s < t \le 1$.
\begin{theorem}\label{prolate_theorem} For every $0 < t < \frac 1 2$,
\begin{equation*} \lim_{n \to \infty} \frac 1 {\alpha n} N_n(0,t) = \lim_{n \to \infty} \frac 1 {\alpha n} N_n(1-t, 1) = \frac 1 2,\end{equation*}
and
\begin{equation*} N_n(t, 1-t) = \mco(\log n).\end{equation*} \end{theorem}
The behavior described in Theorem \ref{prolate_theorem} is typical in the context of Slepian spectral concentration problems (\cite{llwang, EMT, slepian, varah}), which involve pairs of (spectral) projections analogous to $P_{3,\alpha,n}$, $P_{1,\alpha,n}$. Namely, one projection (say, $Q_1$) is the operator of multiplication by an indicator function of a finite interval, and the second ($Q_2$) is obtained from the first through conjugation by the Fourier transform (on $\RR$, $S^1$, or $\ZZ_n = \ZZ / n\ZZ$). Then, the problem is essentially to investigate the spectral properties and the eigenfunctions of $Q_2 Q_1 Q_2$. While numerous variants of this problem have been explored in great detail, we were not able find literature on the specific one addressed here.

Theorem \ref{prolate_theorem} is illustrated for $P_{3,n}$, $P_{1,n}$ in figures \ref{fig7}, \ref{fig8}. In particular, it provides further evidence that pairs of spectral projections of spin operators are unlike random pairs of projections\footnote{Nonetheless, in the case of $P_{1,n}$, $P_{3,n}$, the limits (\ref{grassmann_limit}), (\ref{spin_limit}) coincide for every $f \in \mca$.}. Namely,
\begin{conclusion}\label{random_angles} Applying Theorem \ref{limit_joint} (which is a reformulation of a result from \cite{dumitriupaquette}) to some suitable $F \in C[0,1]$, we see that for every $0 < \varepsilon < \frac 1 2$ it is possible to choose $r > 0$ such that the measure of the set
\begin{equation*} \left\{(P,Q) \in \rmg_{\lfloor \alpha n \rfloor}(n) \times \rmg_{\lfloor \frac 1 2 n \rfloor}(n) \ | \ \#\left(\sigma(P Q P) \cap [\varepsilon, 1- \varepsilon]\right) \ge r n \text{ elements}\right\}\end{equation*}
converges to $1$ as $n \to \infty$. Here, $\sigma(A)$ is the spectrum of a linear operator $A$ and $\# S$ denotes the number of elements of a finite set $S$.\end{conclusion}
\begin{figure}[H]
    \centering
        \includegraphics[width=1\textwidth]{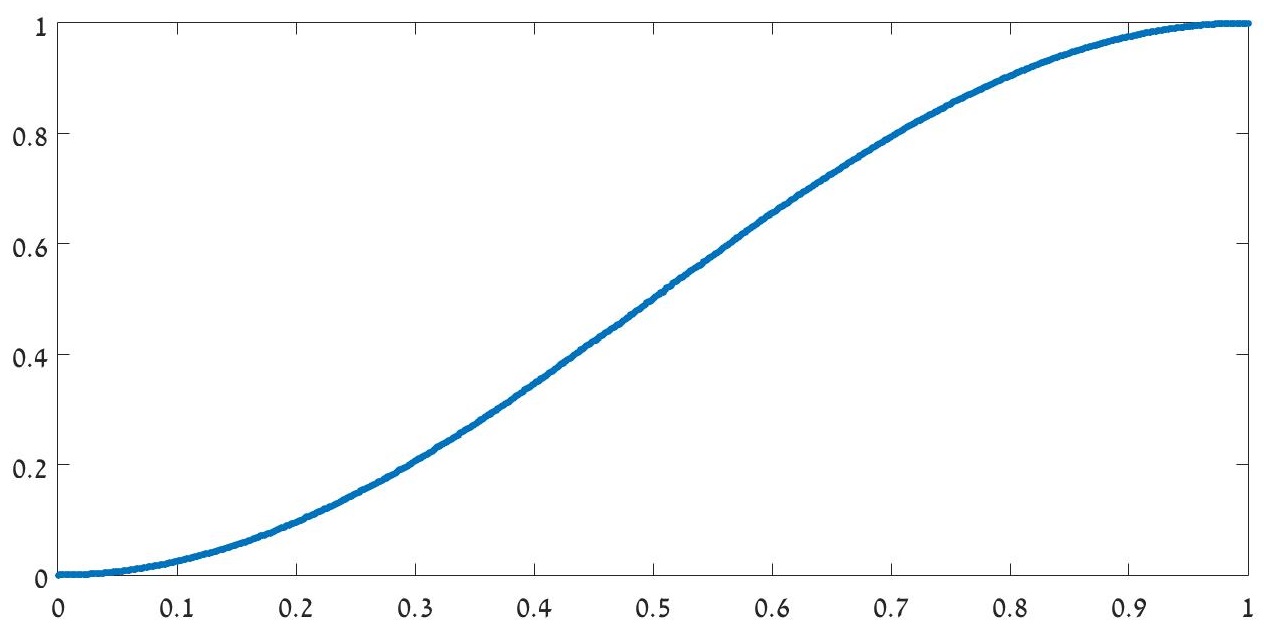}
        \caption{The sorted eigenvalues of $P Q P\in \End\left(V_{P_0} \right)$, where $P,Q \in \rmg_{n}(2n)$ are random and $n = 1000$.}\label{fig7}
\end{figure}
\begin{figure}[H]
    \centering
        \includegraphics[width=1\textwidth]{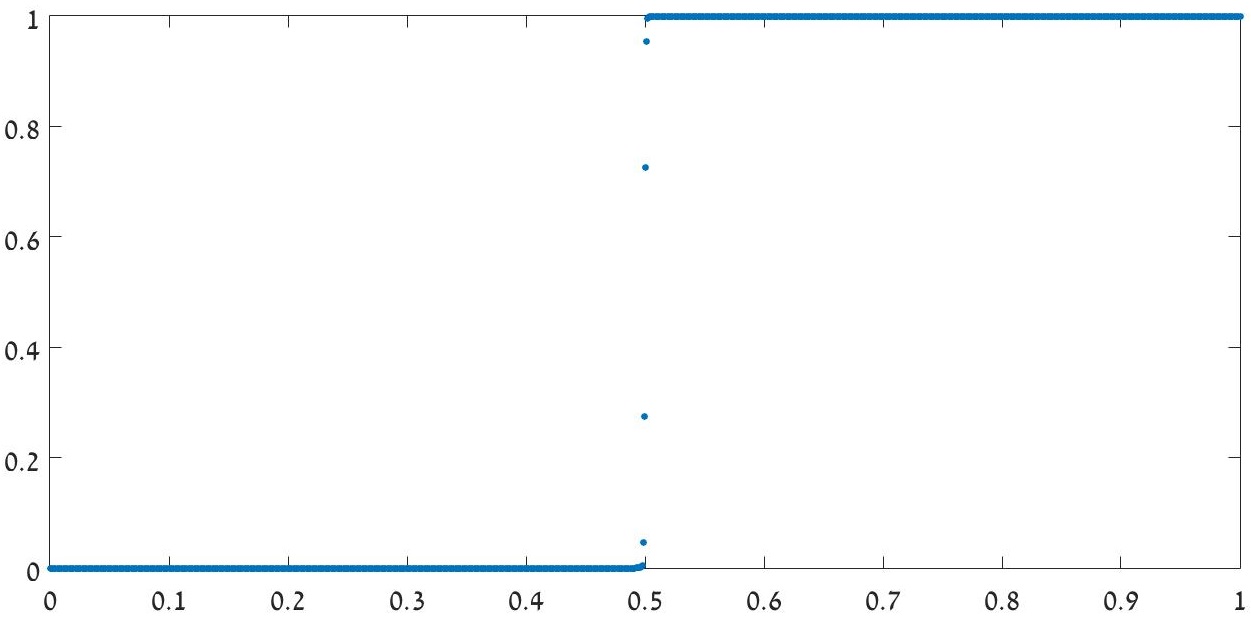}
        \caption{The sorted eigenvalues of $P_{3,n} P_{1,n} P_{3,n}\in \End\left(\text{Im}\left(P_{3,n}\right)\right)$ for $n = 2000$. The number of values that are visibly between $0$ and $1$ is $\mco(\log n)$.}\label{fig8}
\end{figure}
The conclusion of Theorem \ref{spin_theorem} also applies to the pairs $P_{3,\alpha,n}$, $P_{1,n}$ (while the corresponding random pairs satisfy a version of Theorem \ref{grassmann_theorem}, as specified in Theorem \ref{maintheorem}). This can be proven using the arguments of Section \ref{spin_section}, and it implies that $\lim_{n \to \infty} N_n(t,1-t) = \infty$. According to our numerical simulations, a version of Theorem \ref{prolate_theorem} applies to the pairs $P_{3,\alpha,n}$, $P_{1,\alpha,n}$ as well. Unfortunately, it is not clear whether this can be proven using the arguments of Section \ref{prolate_section}.
%
\subsection{Examples}
Every polynomial $f \in \mca$ can be written uniquely as
\begin{equation*} f(x,y) = c_0 + f_1(xy) xy + f_2(yx) yx + f_3(xy) x + f_4(yx) y,\end{equation*}
where $c_0 \in \CC$ and $f_1, f_2, f_3, f_4$ are univariate polynomials. $\psi_f$ admits a rather concise formula (see Theorem \ref{normformula}) in terms of $f_1, f_2, f_3, f_4$. For instance, when $f(x,y) = f_1(xy) xy - f_1(yx) yx$ (see Example \ref{abelianization2}),
\begin{equation*} \psi_f(t) = \lvert f_1(t) \rvert \sqrt{t(1-t)}.\end{equation*}
More specifically,
\begin{example}\label{example_psi_f} Let $f(x,y) = (xy)^{k+1} - (yx)^{k+1}$, where $k \ge 0$. Then
\begin{equation*} \psi_f(t) = t^k \sqrt{t(1-t)},\end{equation*}and
\begin{equation*} M_f = \max_{[0,1]} \psi_f(t) = \psi_f \left(\frac {2k+1} {2k+2} \right) = \frac 1 {\sqrt{2k+2}} \left(1- \frac 1 {2k+2} \right)^{k + \frac 1 2}. \end{equation*}
Thus,
\begin{gather*} I^f_{\alpha} = \max_{\left[0,4\alpha(1-\alpha)\right]} \psi_f =  \left\{\begin{array}{ll} M_f & \frac 1 2 - \frac 1 {2 \sqrt{2k+2}} \le \alpha \\ \left(4\alpha(1-\alpha) \right)^{k + \frac 1 2}(1-2\alpha) & 0 < \alpha < \frac 1 2 - \frac 1 {2 \sqrt{2k+2}} \end{array}\right..\end{gather*} \end{example}

We further specialize and consider $f(x,y) = xy - yx$. The previous example shows that $M_f = \frac 1 2$ and
\begin{equation*} I^f_{\alpha} = \left\{\begin{array}{ll} \frac 1 2 & \frac 1 2 - \frac 1 {2\sqrt{2}} \le \alpha \le \frac 1 2,\\ 2(1-2\alpha)\sqrt{\alpha(1-\alpha)} & 0 < \alpha < \frac 1 2 - \frac 1 {2\sqrt 2} \end{array} \right. \end{equation*}
This is illustrated in the following images.
\begin{figure}[H]
    \centering
        \includegraphics[width=1\textwidth,center]{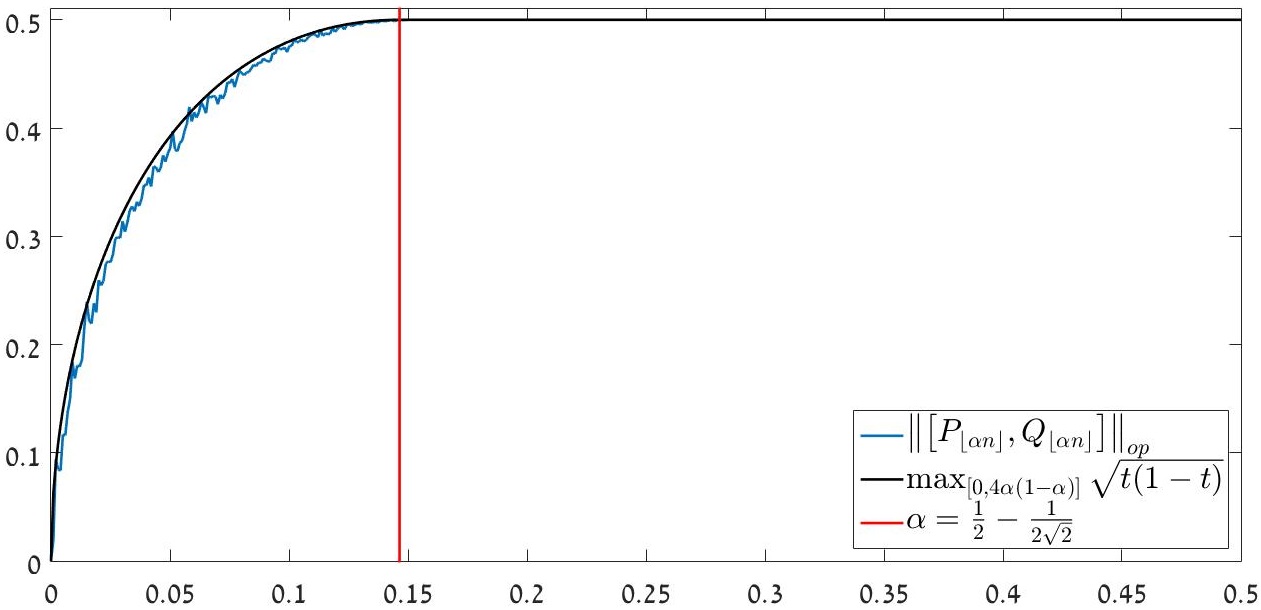}
        \caption{A simulation of randomly drawn $\left \Vert \left [P_{\lfloor \alpha n \rfloor}, Q_{\lfloor \alpha n \rfloor} \right] \right \Vert_{\op}$ as a function of $\alpha$ for $n = 1000$.}\label{fig5}
\end{figure}
The corresponding image for spectral projections of spin operators is (again) very dissimilar.
\begin{figure}[H]
    \centering
        \includegraphics[width=1.0\textwidth, center]{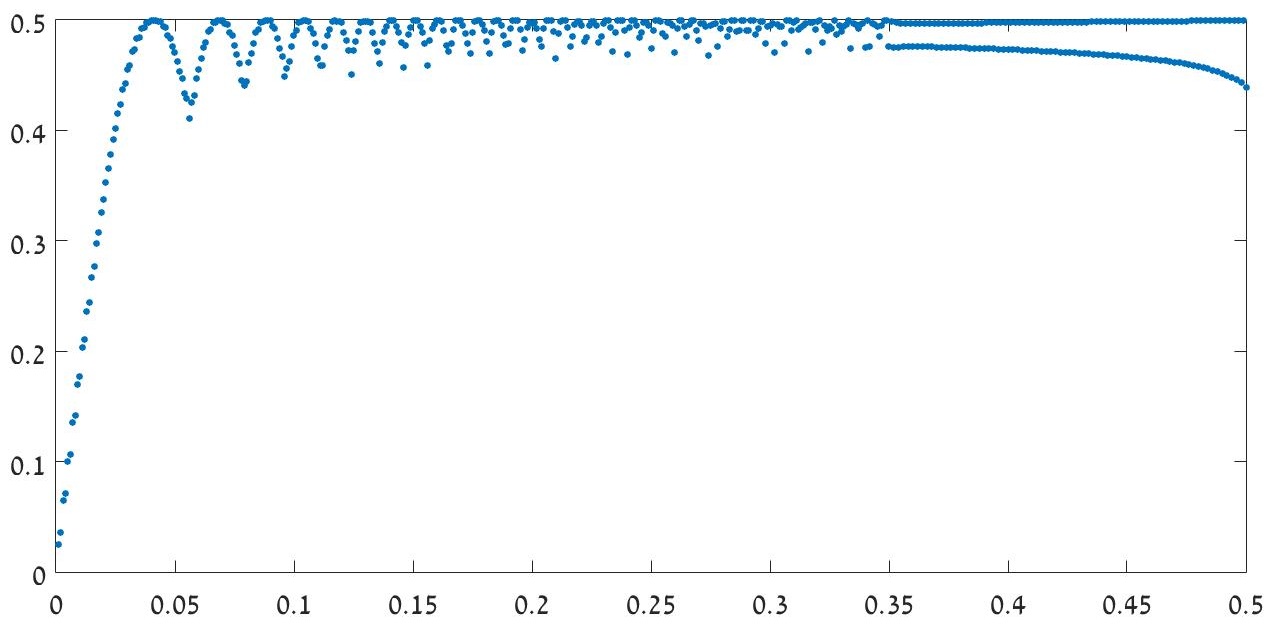}
        \caption{$\left\Vert \left[P_{1,\alpha, n}, P_{3,\alpha,n} \right] \right \Vert_{\op}$ as a function of $\alpha$, where $n = 1000$.}\label{fig6}
\end{figure}
\section{Preliminaries on two orthogonal projections}\label{two_projections_section}
The proofs of Theorems \ref{grassmann_theorem}, \ref{spin_theorem} are independent, but both rely on the general theory of two projections, which we now present. Nearly all of the contents of this section can be found in the excellent guide to the theory \cite{bottspitk}. A few minor modifications and lemmas were added for use in subsequent parts of the work.

Let $\mch$ be a separable complex Hilbert space, possibly infinite dimensional. A pair of orthogonal projections $P : \mch \to V_P,\ Q : \mch \to V_Q$ give rise to a decomposition of $\mch$ as the orthogonal direct sum
\begin{equation}\label{canonic_decomp} \mch = V_{00} \oplus V_{01} \oplus V_{10} \oplus V_{11} \oplus V_0 \oplus V_1,\end{equation}
where
\begin{align*} V_{00} &= V_P \cap V_Q,\ V_{01} = V_P \cap V_Q^\perp,\\ V_{10} &= V_P^\perp \cap V_Q,\ V_{11} = V_P^\perp \cap V_Q^\perp,\\ V_0 &= \left(V_{00} \oplus V_{01} \right)^\perp\cap V_P,\ V_1 = \left(V_{10} \oplus V_{11} \right)^\perp \cap V_P^\perp,\end{align*}
so that
\begin{equation*} V_P = V_{00} \oplus V_{01} \oplus V_0,\ V_P^\perp = V_{10} \oplus V_{11} \oplus V_1.\end{equation*}
Of course, some of the summands in the decomposition (\ref{canonic_decomp}) of $\mch$ may be trivial.
\begin{remark} Let $V = V_0 \oplus V_1$. Then $P, Q$ commute on $V^\perp = \bigoplus_{l,k \in \{0,1\}} V_{lk}$. Hence, unless stated otherwise, we assume throughout this section that $V \ne \{0\}$.\end{remark}
Given $\alpha_{lk} \in \CC$, where $l,k \in \{0,1\}$, we abbreviate
\begin{equation*} (\alpha_{00}, \alpha_{01}, \alpha_{10}, \alpha_{11}) = \alpha_{00} I_{V_{00}} \oplus \alpha_{01} I_{V_{01}} \oplus \alpha_{10} I_{V_{10}} \oplus \alpha_{11} I_{V_{11}} : V^\perp \to V^\perp.\end{equation*}
Clearly
\begin{equation*} P|_{V^\perp} = (1,1,0,0),\ Q|_{V^\perp} = (1,0,1,0)\end{equation*}
and $P|_{V} = \diag(I, 0)$.
A canonical form of the pair $P,Q$ is specified as follows.
\begin{theorem}[\cite{halmos}]\label{canonic_repr} $V_0 \ne \{0\}$ if and only if $V_1 \ne \{0\}$. If $V_0, V_1 \ne \{0\}$, then
\begin{align*} P &= (1,1,0,0) \oplus P|_V = (1,1,0,0) \oplus \left(\begin{array}{cc} I & 0 \\ 0 & 0 \end{array}\right),\\ Q &= (1, 0, 1, 0) \oplus Q|_V = (1,0,1,0) \oplus U^* \left(\begin{array}{cc} I - H & \sqrt{H(I-H)} \\ \sqrt{H(I-H)} & H \end{array}\right)U.\end{align*}
Here $0\le H\le 1$ with $\ker H = \ker(I-H) = \{0\}$ and $U =\diag \left(I, R \right)$ with $R : V_1 \to V_0$ unitary. \end{theorem}
In particular, we note that $P,Q$ are non-commuting if and only if $V_0 \ne \{0\}$.
\subsection{Polynomials in two non-commuting idempotents}
Recall that $\mca$ denotes the complex algebra generated by two non-commuting variables $x,y$ which satisfy the relations $x= x^2$, $y =y^2$. A basis of $\mca$ as a vector space is provided by the monomials
\begin{equation}\label{monomialsbasis} 1,\ (xy)^{k+1},\ (yx)^{k+1},\ (xy)^k x,\ (yx)^k y,\end{equation}
where $k \ge 0$. Thus any $f \in \mca$ decomposes uniquely as
\begin{equation*} f(x,y) = a_0 + f_1(xy)xy+ f_2(yx)yx + f_3(xy) x +f_4(yx) y,\end{equation*}
where $f_1, f_2, f_3, f_4$ are complex univariate polynomials. Let $r(t) = \sqrt{t(1-t)}$, and for $l,k \in \{0,1\}$ define $g_{lk} : [0,1] \to \CC$ by
\begin{align*}
g_{00}(t) &= a_0 + f_3(t) +t  \left( f_1(t) + f_2(t) + f_4(t)\right),\\ g_{01}(t) &= r(t)\left(f_1(t) + f_4(t)\right),\ g_{10}(t) = r(t)\left(f_2(t) + f_4(t)\right),\\ g_{11}(t) &= a_0  + (1-t)f_4(t).\end{align*}
Also, denote
\begin{equation*} \alpha_{00} = g_{00}(1),\ \alpha_{01} = g_{00}(0),\ \alpha_{10} = g_{11}(0),\ \alpha_{11} = g_{11}(1).\end{equation*}

\begin{lemma}\label{abelianization} Let $T : \mca \to \CC[z, w]\left/ \mci \right.$ be the "abelianization" map of (\ref{abelianizationmap}). Then $f \in \ker (T)$ if and only if
\begin{equation*} \alpha_{00} = \alpha_{10} = \alpha_{01} = \alpha_{11} = 0.\end{equation*} \end{lemma}
\begin{proof}
Write
\begin{equation*} f_l(t) = \sum_{k=0}^{k_l} a_k^{(l)} t^k,\ l = 1,2,3,4.\end{equation*}
A straightforward computation shows that
\begin{gather*} \left(Tf\right)(z,w) =\\ a_0 + f_3(0) z + f_4(0) w + \left[f_1(1) + f_2(1) + f_3(1) - f_3(0) + f_4(1) - f_4(0) \right] zw.\end{gather*}
Thus we obtain the required. \end{proof}

Next, by induction,
\begin{align*}(PQ)^{k+1} &= I_{V_{00}}\oplus U^* \left(\begin{array}{ll} (I-H)^{k+1} & (I-H)^k r(H) \\ 0 & 0 \end{array}\right) U,\\ (QP)^{k+1} &= I_{V_{00}}\oplus U^* \left(\begin{array}{ll} (I-H)^{k+1} & 0 \\ (I-H)^k r(H) & 0 \end{array}\right) U,\\ (PQ)^{k}P &=I_{V_{00}}\oplus U^* \left(\begin{array}{ll} (I-H)^{k} & 0 \\ 0 & 0 \end{array}\right) U,\\ (QP)^{k}Q &= I_{V_{00}}\oplus U^*  \left(\begin{array}{ll} (I-H)^{k+1} & (I-H)^{k} r(H) \\ (I-H)^{k} r(H) & (I-H)^{k} H \end{array}\right)U,\end{align*}
where $k \ge 0$. By linearity, we obtain a precise expression for $f(P,Q)$ as follows. 
\begin{conclusion}[\cite{gileskummer}]\label{canonic_poly} For every complex Hilbert space $\mch$ and for any pair of non-commuting orthogonal projections $P : \mch \to V_P,\ Q : \mch \to V_Q$,
\begin{gather*} f(P,Q) = a_0 + f_1(PQ)PQ + f_2(QP)QP + f_3(PQ)P + f_4(QP)Q=\\ (\alpha_{00}, \alpha_{01}, \alpha_{10}, \alpha_{11}) \oplus U^* \left(\begin{array}{ll} g_{00}(I-H) & g_{01}(I-H) \\ g_{10}(I-H) & g_{11}(I-H)\end{array}\right) U.\end{gather*}\end{conclusion}
\begin{proof}
Denote $f_l(t) = \sum_{k=0}^{k_l} a_k^{(l)} t^k$, where $l=1,2,3,4$ and $k_l \ne 0$. We compute each of the summands $f_1, f_2, f_3, f_4$ of $f(P,Q)$ separately.

The first summand is
\begin{gather*} f_1(PQ)PQ = \sum_{k=0}^{k_1} a_k^{(1)} (PQ)^{k+1} = \\\left( f_1(1) I_{V_{00}}\right) \oplus U^* \left(\begin{array}{ll} f_1(I-H)(I-H) & f_1(I-H) r(H) \\ 0 & 0 \end{array}\right) U. \end{gather*}

The second summand is
\begin{gather*} f_2(QP)QP = \sum_{k=0}^{k_2} a_k^{(2)}(QP)^{k+1} = \\ \left(f_2(1) I_{V_{00}} \right) \oplus U^* \left(\begin{array}{ll} f_2(I-H)(I-H) & 0 \\ f_2(I-H) r(H) & 0 \end{array}\right)U.\end{gather*}

The third summand is
\begin{equation*} f_3(PQ)P = f_3(PQP)P = a_0^{(3)} P + \sum_{k=0}^{k_3-1} a_{k+1}^{(3)} (PQP)^{k+1},\end{equation*}
where
\begin{gather*} \sum_{k=0}^{k_3-1} a_{k+1}^{(3)} (PQP)^{k+1} = \\ \left((f_3(1)- a_0^{(3)}) I_{V_{00}} \right)\oplus U^* \left(\begin{array}{cc} \sum_{k=0}^{k_3 - 1} a_{k+1}^{(3)} (I-H)^{k+1} & 0 \\ 0 & 0 \end{array}\right)U,\end{gather*}
so
\begin{gather*} f_3(PQ)P = (f_3(1), f_3(0), 0, 0) \oplus U^* \left(\begin{array}{ll} f_3(I-H) & 0 \\ 0 & 0 \end{array}\right) U.\end{gather*}

Finally,
\begin{equation*} f_4(QP)Q = f_4(QPQ)Q = a_0^{(4)} Q +\sum_{k=0}^{k_4-1} a_{k+1}^{(4)}(QPQ)^{k+1},\end{equation*}
so
\begin{gather*} f_4(QP)Q =\\ (f_4(1), 0, f_4(0), 0) \oplus U^* \left(\begin{array}{ll} f_4(I-H)(I-H) & f_4(I-H) r(H) \\ f_4(I-H)r(H) & f_4(I-H) H \end{array}\right)U.\end{gather*}
Putting everything together, we obtain the required.
\end{proof}

\begin{remark}\label{free_abelianization} We can factor $T$ as $ \Pi_{\mci} \circ T_0$, where $T_0 : \mca \to \CC[z,w]$ is the linear map $f(x,y) \mapsto f(z,w)$ (unlike $T$, the map $T_0$ is not a morphism of algebras) and $\Pi_{\mci}$ is the projection onto $\CC[z,w]/ \mci$. Note that $f \in \ker (T_0)$ if and only if
\begin{equation*} f(x,y) = f_1(xy) xy - f_1(yx) yx = \left[x, f_1(yx) y \right],\end{equation*} where $f_1$ is a complex univariate polynomial, and then
\begin{gather*} f(P,Q) = \left[P, f_1(QP)Q \right] = \\ = (0,0,0,0) \oplus U^* \left(\begin{array}{cc} 0 & f_1(I-H) r(H) \\ - f_1(I-H) r(H) & 0 \end{array}\right)U.\end{gather*}
\end{remark}
\subsection{The operator norm of polynomials in two projections}\label{norm_subsection}
We define (following \cite{spitkovsky}) two functions $g_0, g_1 \in C \left([0,1], \CC \right)$ by
\begin{equation}\label{canonic_funcs} g_0 = \sum_{l,k = 0}^1 \lvert g_{lk} \rvert^2,\ g_1 = g_{00} g_{11} - g_{01} g_{10},\end{equation}
where $g_{00}, g_{01}, g_{10}, g_{11}$ are the functions appearing in Conclusion \ref{canonic_poly}.
\begin{remark} In the literature, it appears that $g_{00}, g_{01}, g_{10}, g_{11}$, and $g_0, g_1$ are considered only as functions on $\sigma(I-H)$. For us, it is very useful that they are determined by $f$ as continuous functions on $[0,1]$, so as to apply by restriction to all separable Hilbert spaces and orthogonal projections.\end{remark}
$g_0, g_1$ can be used to express $\sigma \left(f(P,Q)\right)$ in terms of $\sigma(I-H)$, since $\left. f(P,Q)\right|_V$ has a bounded inverse if and only if $g_1(I-H)$ does. In particular, the polynomial $f(P,Q) f^*(P,Q)$ can be used to determine $\Vert f(P,Q) \Vert_{\op}$.
\begin{theorem}[\cite{bottspitk, spitkovsky}]\label{normformula} Let $\Lambda = \{(l,k) \in \{0,1\}^2 \ | \ V_{lk} \ne \{0\}\}$. The operator norm of $f(P,Q)$ is given by
\begin{equation*} \Vert f(P,Q) \Vert_{\op} = 
  \max\left\{\max_{(l,k) \in \Lambda} \lvert \alpha_{l,k} \rvert, \max_{t \in \sigma(I-H)} \psi_f(t)\right\},\end{equation*}
where $\psi_f\in C([0,1], \RR)$ is given by
\begin{equation*} \psi_f(t) = \sqrt{\frac{g_0(t) + \sqrt{g_0^2(t) - 4 \lvert g_1(t) \rvert^2}} 2 }.\end{equation*}\end{theorem}
Note that
\begin{equation*} \psi_f(0) = \max\{|\alpha_{10}|, |\alpha_{01}|\},\ \psi_f(1) = \max\{|\alpha_{11}|, |\alpha_{00}|\}.\end{equation*}
Assume $f \in \ker(T)$, so that by Lemma \ref{abelianization}, $\alpha_{lk} = 0$ for $l,k = 0,1$, or equivalently $\psi_f(0) = \psi_f(1)=0$. If $P,Q$ commute, then $f(P,Q) = 0$ and $\max_{\sigma(PQP)}\psi_f = 0$. Otherwise if $P,Q$ do not commute, $\max_{\sigma(I-H)}\psi_f = \max_{\sigma(PQP)} \psi_f$. We use this to reformulate Theorem \ref{normformula} as follows.
\begin{conclusion}\label{reformulated_norm_formula} Let $\Lambda = \{(l,k) \in \{0,1\}^2 \ | \ V_{lk} \ne \{0\}\}$ as above, and $f \in \mca$.
\begin{enumerate}
\item{Assume $f \in \ker (T)$. Then for every separable complex Hilbert space $\mch$ and orthogonal projections $P,Q$ on $\mch$,
\begin{equation*} \Vert f(P,Q) \Vert_{\op} = \max_{\sigma(PQP)} \psi_f.\end{equation*}}
\item{Define $\Psi_f : [0,1] \cup \{0,1\}^2 \to \RR$ by $\left. \Psi_f \right|_{[0,1]} = \psi_f$ and $\Psi_f((l,k)) = |\alpha_{lk}|$ for $(l,k) \in \{0,1\}^2$. Denote $\sigma_{P,Q} = \sigma(I-H) \cup \Lambda$. Then for every separable complex Hilbert space $\mch$ and orthogonal projections $P,Q$ on $\mch$,
\begin{equation*} \Vert f(P,Q) \Vert_{\op} = \max_{\sigma_{P,Q}} \Psi_f.\end{equation*}}
\end{enumerate}
Note that $\max_{[0,1]} \psi_f = \max_{[0,1] \cup \{0,1\}^2} \Psi_f$, since $\psi_f(1) = \max\{|\alpha_{00}|, |\alpha_{11}|\}$ and $\psi_f(0) = \max\{|\alpha_{01}|, |\alpha_{10}|\}$.
\end{conclusion}
The latter, together with Claim \ref{max_norm}, immediately leads to the following.
\begin{conclusion}\label{universal_upper_bound} The constant
\begin{equation*} M_f = \max_{[0,1]} \psi_f = \max_{[0,1] \cup \{0,1\}^2}\Psi_f\end{equation*} is a universal, tight upper bound for $\Vert f(P,Q) \Vert_{\op}$, where $P,Q$ are any orthogonal projections on an arbitrary complex Hilbert space $\mch$.
\end{conclusion}
We conclude with the following example.
\begin{example}\label{abelianization2} Let $f \in \ker (T_0)$ where $T_0$ is as in Remark \ref{free_abelianization}, and recall that $r(t) = \sqrt{t(1-t)}$. Then
\begin{equation*} f(P,Q) = f_1(PQ)PQ - f_1(QP)QP\end{equation*}
for some univariate polynomial $f_1$, therefore
\begin{equation*} g_{00} = g_{11} = 0,\ g_{01} = -g_{10}(t) = f_1(t) r(t).\end{equation*}
It follows that $g_0(t) = 2 \lvert g_{01}(t) \rvert^2,\ g_1(t) = g_{01}(t)^2$, hence we find that
\begin{gather*} \psi_f(t) = \sqrt{\frac{g_0(t) + \sqrt{g_0(t)^2 - 4 \lvert g_1(t) \rvert^2}} 2 } \\ = \sqrt{\frac{2 \lvert g_{01}(t) \rvert^2 + \sqrt{4 \lvert g_{01}(t) \rvert^4 - 4 \lvert g_{01}(t) \rvert^4}} 2 } = \lvert g_{01} (t) \rvert.\end{gather*}
Thus,
\begin{equation*} \Vert f(P,Q) \Vert_{\op} = \max_{ \sigma(I-H)}\left( \lvert f_1 \rvert r\right) \le \max_{[0,1]} \left( \lvert f_1 \rvert r\right) = M_f.\end{equation*}.
\end{example}
\subsection{The canonical form and angles between subspaces}\label{angles_subsection}
Assume that $\dim \mch = n$. Recall the notation
\begin{align*} V_{00} &= V_P \cap V_Q,\ V_{01} = V_P \cap V_Q^\perp,\\ V_{10} &= V_P^\perp \cap V_Q,\ V_{11} = V_P^\perp \cap V_Q^\perp\\ V_0 &= \left(V_{00} \oplus V_{01} \right)^\perp\cap V_P,\ V_1 = \left(V_{10} \oplus V_{11} \right)^\perp \cap V_P^\perp,\end{align*}
and let $m_{lk} = \dim V_{lk},\ l,k\in \{0,1\}$ and $m = \dim V_0 = \dim V_1$.
\begin{definition}\label{reduced_principal_angles} Denote the eigenvalues of $H$ by $0 < \mu_1 \le ... \le \mu_m < 1$. The \textit{reduced} principal angles $0 <\theta_1 \le ... \le \theta_m < \frac \pi 2$ associated with the pair $(P,Q)$ are defined by
\begin{equation*} \sin^2 \theta_l = \mu_l,\ l = 1,...,m.\end{equation*} \end{definition}
The pair $(P,Q)$ is determined, up to unitary equivalence, by the numbers $m_{00}, m_{01}, m_{10}, m_{11}, m$ together with the reduced principal angles.
\begin{definition} Let $m_P = \dim V_P \le \dim V_Q = m_Q$ (so $m_P = m_{00} + m_{01} + m$). The principal angles $0 \le \phi_1 \le \phi_2 \le... \le \phi_{m_P} \le \frac \pi 2$ of the pair $(P, Q)$ are defined recursively. The angle $\phi_1$ is specified by
\begin{equation*} \cos \phi_1 = \max\{|\langle x, y \rangle | \ | \ x \in V_P, y \in V_Q, \Vert x \Vert = \Vert y \Vert = 1\},\end{equation*}
and if $\cos \phi_1 = |\langle x_1, y_1 \rangle |$, then
\begin{equation*} \cos \phi_2 = \max\{|\langle x, y \rangle | \ | \ x \in V_P \cap \{x_1\}^\perp,\ y \in V_Q \cap \{y_1\}^\perp,\ \Vert x \Vert = \Vert y \Vert = 1\}.\end{equation*}
Next, if we denote $V_{P,k} = \Sp\{x_1, ..., x_k\}$ and $V_{Q,k} = \Sp\{y_1, ..., y_k\}$, where $\cos \varphi_l = |\langle x_l, y_l \rangle |$ for $l = 1,...,k$, then
\begin{equation*} \cos \phi_{k+1} = \max \{|\langle x, y \rangle | \ | \ x \in V_P \cap V_{P,k}^\perp,\ y \in V_Q \cap V_{Q,k}^\perp,\ \Vert x \Vert = \Vert y \Vert = 1\}.\end{equation*} \end{definition}
The reduced principal angles are the principal angles lying in $\left(0, \frac \pi 2 \right)$, i.e.,
\begin{gather*} \phi_1 = ... = \phi_{m_{00}} = 0,\\ \phi_{m_{00}+1} = \theta_1, ..., \phi_{m_{00} + m} = \theta_m,\\ \phi_{m_{00} + m + 1} = ... = \phi_{m_{00} + m + m_{01}} = \frac \pi 2.\end{gather*}
The following elementary examples will be of immediate use.
\begin{example} Let $n \ge 2$ and $\phi \in \left[0, \frac \pi 2 \right]$. For $0 < m_P \le m_Q \le n - m_P$ integers, there exists a pair of orthogonal projections $P : \CC^n \to V_P$, $Q : \CC^n \to V_Q$ with $m_P = \dim V_P,\ m_Q = \dim V_Q$ and smallest principal angle $\phi_1 = \phi$.

Indeed, if $\{e_l\ | \ l = 1,...,n\}$ is an orthonormal basis of $\CC^n$, and
\begin{equation*} V_P = \Sp\left\{e_1, ..., e_{m_P}\right\},\end{equation*}
then
\begin{equation*} V_Q = \Sp\left\{\cos \phi e_1 + \sin \phi e_n,\ e_{m_P + 1}, ..., e_{m_P + m_Q - 1} \right\}\end{equation*}
satisfies the required. In particular, if $\phi \in \left(0, \frac \pi 2 \right)$, then $\sin^2 \phi \in \sigma(H)$.\end{example}
\begin{example} Let $n \ge 2$, $\phi \in \left[0, \frac \pi 2 \right]$ and assume that $0 < m_P \le m_Q < n$. Then there exists a pair of orthogonal projections $P : \CC^n \to V_P, Q : \CC^n \to V_Q$ with $m_P = \dim V_P$ and $m_Q = \dim V_Q$ such that $\phi \in \{\phi_1, ..., \phi_{m_P}\}$.

Indeed, if $m_P + m_Q \le n$, we saw that it is possible to define $P,Q$ such that $\phi = \phi_1$. If $m_P + m_Q > n$, then $m_{00} = \dim V_P \cap V_Q > 0$, so $\phi_1 = ... = \phi_{m_{00}} = 0$. Thus if $\phi = 0$, we are done. Otherwise, we may set $m_{00} = m_P + m_Q - n$ (note that $m_{00} < m_P$), $V_P = \Sp\{e_1, ..., e_{m_P}\}$ and
\begin{equation*} V_Q = \Sp\{e_1,...,e_{m_{00}}, \cos \phi e_{m_P} + \sin \phi e_n, e_{m_P+1}, ..., e_{n-1}\}.\end{equation*}
Then $P,Q$ are as required.\end{example}
The previous examples essentially amount to the proof of the following, where $M_f$ is the universal bound of Conclusion \ref{universal_upper_bound}.
\begin{claim}\label{max_norm} Let $\mch$ be a separable complex Hilbert space with $2 \le \dim \mch \le \infty$. Let $0, \Id \ne P : \mch \to V_P$ denote an orthogonal projection. Then for every $1 \le m_Q \le \dim \mch-1$ there exists an orthogonal projection $Q : \mch \to V_Q$ with $\rk Q = m_Q$ such that $\Vert f(P,Q) \Vert_{\op} = M_f$.\end{claim}
\begin{proof}
We need to find an orthogonal projection $Q : \mch \to V_Q$ such that
\begin{equation*} \Vert f(P,Q) \Vert_{\op} = M_f.\end{equation*}

We assume without loss of generality that $\dim \mch = 2$, so that $\dim V_P = 1$. Indeed, if we choose arbitrary non-zero $v_1 \in V_P$ and $v_2 \in \ker (P)$, and denote $\mch_0 = \Sp\{v_1, v_2\}$, and find an orthogonal projection $Q : \mch_0 \to \mch_0$ such that $\Vert f(P|_{\mch_0}, Q) \Vert_{\op} = \max_{[0,1]} \psi_f$, then clearly we can extend $Q$ to $\mch_0^\perp$ so that $\rk Q = m_Q$, which yields the required. Hence we assume $\dim \mch = 2$.

Let $t_{\max} \in [0,1]$ be such that $\max_{[0,1]} \psi_f(t) = \psi_f(t_{\max})$. If $0 < t_{\max} < 1$, then there exists $\phi \in \left(0, \frac \pi 2 \right)$ such that $\sin^2 \phi = 1-t_{\max}$. Clearly, the previous examples imply that there exists $Q : \mch \to V_Q$ such that $\phi$ is among the (reduced) principal angles associated with $P,Q$, which means that $t_{\max} \in \sigma(I-H)$. Hence, $\Vert f(P,Q) \Vert_{\op} = \max_{[0,1]} \psi_f$.

Otherwise, $t_{\max}\in \{0,1\}$, which means that
\begin{equation*} \max_{[0,1]} \psi_f = \max_{[0,1]\cup \{0,1\}^2} \Psi_f = \max\{|\alpha_{00}|, |\alpha_{01}|, |\alpha_{10}|, |\alpha_{11}|\}.\end{equation*}
If the latter equals $|\alpha_{00}|$ or $|\alpha_{11}|$, we may set $Q = P$, so that $V_P \cap V_Q = V_P$ and $V_P^\perp \cap V_Q^\perp = V_P^\perp$. Then $\Vert f(P,Q) \Vert_{\op} = \max\{|\alpha_{00}|, |\alpha_{11}|\} = \max_{[0,1]} \psi_f$. Otherwise, we can set $Q = I - P$, so that $V_P \cap V_Q^\perp = V_P,\ V_P^\perp \cap V_Q = V_P^\perp$ to obtain the required.
\end{proof}
\section{Proof of Theorem \ref{grassmann_theorem}}\label{grassmann_section}
The present section is dedicated to the proof of Theorem \ref{maintheorem}, which includes Theorem \ref{grassmann_theorem} as a special case. The proof consists of little more than a straightforward combination of results from \cite{bottspitk} and \cite{collins}. Indeed, if $f \in \ker(T)$, then by Conclusion \ref{reformulated_norm_formula}, $\left \Vert f(P,Q)\right \Vert_{\op}= \max_{\sigma(PQP)} \psi_f$ (the general case $f \in \mca$ is only slightly more complicated). The operator $PQP \in \End(V_P)$ is known to provide a model of the so called Jacobi ensemble of random matrices. In particular, the behavior of $\sigma(PQP)$ is well understood as $n \to \infty$.

Let $\Omega_{n,\alpha,\beta} = \rmg_{a_n}(n) \times \rmg_{b_n}(n)$, equipped with the probability measure $\nu_{n,\alpha,\beta} = \mu_{a_n, n} \times \mu_{b_n,n}$. Here, $a_n = \lfloor \alpha n \rfloor$ and $b_n = \lfloor \beta n \rfloor$, where $0 < \alpha \le \beta$ satisfy $\alpha + \beta \le 1$. For $(P, Q) \in \Omega_{n,\alpha,\beta}$, we will use the notations of (\ref{canonic_decomp}). Finally, let
\begin{equation*} I^{f}_{n,\alpha,\beta} = \int_{\Omega_{n,\alpha,\beta}} \Vert f(P, Q) \Vert_{\op} d\nu_{n,\alpha,\beta}.\end{equation*}

Our goal is to compute $I^f_{\alpha, \beta} = \lim_{n\to \infty} I^{f}_{n, \alpha, \beta}$ in terms of $\alpha, \beta$ and $f$, where according to Fubini's theorem,
\begin{equation}\label{fubini} I^{f}_{n,\alpha, \beta} =\int_{\rmg_{a_n}(n)}  \left( \int_{\rmg_{b_n}(n)} \Vert f(P, Q) \Vert_{\op} d\mu_{b_n,n}(Q) \right) d\mu_{a_n, n}(P). \end{equation}
In this context, we note the following observation.
\begin{remark} $\Vert f(P, Q) \Vert_{\op} \le M_f$ for all $(P, Q) \in \Omega_{n,\alpha,\beta}$ by Conclusion \ref{universal_upper_bound}, hence $I^{f}_{n, \alpha, \beta} \le M_f$.
Also note that $\max_{Q \in \rmg_{b_n}(n)} \Vert f(P, Q) \Vert_{\op} = M_f$ for every $P \in \rmg_{a_n}(n)$, by Claim \ref{max_norm}.
\end{remark}

The invariance of $\mu_{b_n, n}$ implies that the inner integral in (\ref{fubini}) is independent of the choice of $P$. To see this, let $\mce_n = \{e_1, ..., e_n\}$ denote the standard basis of $\CC^n$, and let $P_{0}$ denote the orthogonal projection on $ \Sp\{e_1, ..., e_{a_n}\}$. Then for every $P \in \rmg_{a_n}(n)$ there exists a (non-unique) unitary operator $U_0$ such that
\begin{equation*} P = U_0^* P_{0} U_0.\end{equation*}
\begin{lemma} Let $Q \in \rmg_{b_n}(n)$. Then $\left \Vert f(P, Q) \right \Vert_{\op} = \left\Vert f(P_{0}, U_0 Q U_0^*)\right \Vert_{\op}$. \end{lemma}
\begin{proof}
It will be convenient to introduce the evaluation homomorphisms
\begin{equation*} \ev(P, Q) : \mca \to \End(\CC^n),\end{equation*} specified by $\ev(P, Q)(f) = f(P, Q)$. Then
\begin{equation*} \ev(P, Q)(\cdot) = U^*_0 \ev(P_{0}, U_0 Q U_0^*)(\cdot) U_0.\end{equation*} Indeed,
\begin{gather*} \ev(P, Q)(1) = I = U_0^* \ev(P_{0}, U_0 Q U_0^*)(1) U_0,\\ \ev(P, Q)(x) = P = U_0^* P_{0} U_0 = U_0^* \ev(P_{0}, U_0 Q U_0^*)(x) U_0\end{gather*}
and
\begin{equation*} \ev(P, Q)(y) = Q = U_0^* \left(U_0 Q U_0^* \right) U_0 = U_0^* \ev\left(P_{0}, U_0 Q U_0^* \right)(y) U_0.\end{equation*}
This holds, similarly, for all the monomials in $\mca$, and by linearity, extends to all of $\mca$. Since $U_0$ is unitary, we obtain the required.
\end{proof}
The invariance of $\mu_{b_n, n}$ implies that
\begin{equation*} \int_{\rmg_{b_n}(n)} \Vert f(P_{0}, U Q U^*) \Vert_{\op} d\mu_{b_n,n}(Q) = \int_{\rmg_{b_n}(n)} \Vert f(P_{0}, Q) \Vert_{\op} d\mu_{b_n,n}(Q)\end{equation*}
for every unitary operator $U$ on $\CC^n$, hence we can conclude that
\begin{equation}\label{simplifiedintegral} I^{f}_{n, \alpha, \beta} = \int_{\rmg_{b_n}(n)} \Vert f(P_{0}, Q) \Vert_{\op} d\mu_{b_n,n}(Q).\end{equation}

Let $H_{0}$ be the operator associated with the pair $(P_{0}, Q)$ as in Theorem \ref{canonic_repr}.
According to Theorem \ref{normformula}, $\Vert f(P_{0}, Q) \Vert_{\op}$ can be expressed conveniently using $\sigma(I - H_{0}) \subset \sigma(P_{0} Q P_{0})$ (assuming that $P_0$, $Q$ do not commute). Thus, we are led to consider the joint distribution of the eigenvalues of $P_{0} Q P_{0}$.
\begin{theorem}[\cite{dumitriupaquette}] The joint eigenvalue distribution of $P_{0} Q P_{0} \in \End(V_{P_{0}})$ in $[0,1]^{a_n}$ is given by
\begin{equation*} d\mcj_n(\lambda_1, ..., \lambda_{a_n}) = \frac 1 {C_{n,\alpha,\beta}} \prod_{l=1}^{a_n} \lambda_l^{b_n - a_n} (1-\lambda_l)^{n -(a_n + b_n)} \prod_{1 \le l < k \le a_n} (\lambda_l - \lambda_k)^2 d\lambda, \end{equation*}
where $C_{n,\alpha,\beta}$ is a normalization constant (and $d\lambda = d\lambda_1 \dots d\lambda_{a_n}$). It is the joint eigenvalue distribution of the (unitary) Jacobi ensemble with parameters $N = a_n, N_1 = b_n, N_2 = n - b_n$ (we adopt the convention of \cite{dumitriupaquette}).\end{theorem}
As an element of $\End(V_{P_{0}})$ (using the notation of (\ref{canonic_decomp})),
\begin{equation*} P_{0} Q P_{0} = I_{V_{00}} \oplus 0_{V_{01}} \oplus \left(I - H_{0} \right).\end{equation*}
Denote
\begin{equation*} G_{n,0} = \{Q \in \rmg_{b_n}(n) \ | \ \dim V_{00} = \dim V_{01} = 0\}.\end{equation*}
Then for $n$ sufficiently large\footnote{Any $n$ works if $\alpha = \beta$, otherwise any $n$ for which $(\beta - \alpha)n \ge 1$ suffices.},
\begin{equation}\label{generic_Q}\mu_{b_n,n}(G_{n,0}) = 1,\end{equation} since by assumption $a_n + b_n \le n$ and $b_n - a_n \ge 0$.
The limiting behavior of $d \mcj_n$ that is relevant in our context is as follows.
\begin{theorem}[\cite{dumitriupaquette}]\label{limit_joint} If $d\mcj_n$ is the joint distribution of $(\lambda_1, ..., \lambda_{a_n}) \in [0,1]^{a_n}$, then for $F \in C[0,1]$,
\begin{equation*} \frac 1 {a_n} \sum_{l=l}^{a_n} F(\lambda_l) \xrightarrow{\mathbb P} \int_0^1 F(t) d\mcj_\infty(t)\end{equation*}
as $n \to \infty$, where
\begin{equation*} d\mcj_\infty(t) = \frac 1 {2\pi \alpha} \frac 1 {t(1-t)} \sqrt{-(t-\lambda_-)(t-\lambda_+)} \mathbbm1_{[\lambda_-, \lambda_+]} dt,\end{equation*}
and $\lambda_\pm = \left( \sqrt{\beta(1-\alpha)} \pm \sqrt{\alpha(1-\beta)} \right)^2$. Note that $\lambda_- < \lambda_+$, and $ \lambda_- = 0$ if and only if $\alpha = \beta$, and $\lambda_+ = 1$ if and only if $\alpha + \beta = 1$.\end{theorem}
The latter does not rule out the possibility, when $0 < \lambda_-$ or $\lambda_+ < 1$, that a subset of order $o(a_n)$ of eigenvalues of $P_{0} Q P_{0}$ remains outside of $[\lambda_-, \lambda_+]$, and a-priori it could be that
\begin{equation*} \mu_{b_n,n} \left(\left\{Q\ | \ \sigma(P_{0} Q P_{0}) \cap \left([0,1] \setminus (\lambda_--\delta, \lambda_+ + \delta)\right) \ne \emptyset\right\} \right) > \varepsilon\end{equation*} for some $\varepsilon, \delta >0$. Hence, we will use the following.
\begin{theorem}[\cite{collins}]\label{spec_concentrate} For any compact set $K$ such that $K \cap [\lambda_-, \lambda_+] = \emptyset$, there exists $C > 0$ such that $\mu_{b_n,n} \left(\left\{Q \ | \ \sigma(P_{0} Q P_{0}) \cap K \ne \emptyset \right\}\right) < e^{-C n}$. \end{theorem}
We conclude as follows.
\begin{conclusion}\label{spec_conclusion} Let $\delta > 0$. Denote\footnote{Here, $d(K, t)$ denotes the distance between the set $K \subset \RR$ and $t \in \RR$.}
\begin{gather*} A_{n,t, \delta} = \{Q \in \rmg_{b_n}(n)\ | \ d(\sigma(I - H_0), t ) < \delta \}, \\B_{n, \delta} = \{Q \in \rmg_{b_n}(n) \ | \ \sigma(I - H_0) \subset (\lambda_- -\delta, \lambda_+ + \delta) \}.\end{gather*}
Then $\lim_{n \to \infty} \mu_{b_n,n} (B_{n, \delta}) = \lim_{n \to \infty} \mu_{b_n,n}(A_{n,t,\delta}) = 1$. \end{conclusion}
\begin{proof}
As we noted previously (\ref{generic_Q}), we can assume that $\mu_{b_n,n}(G_{n,0}) =1$. Thus, almost surely $\sigma(P_0 Q P_0) = \sigma(I - H_0)$, so it is a basic consequence of the fact that the density of $d \mcj_\infty$ is non-vanishing on $[\lambda_-, \lambda_+]$ that $\mu_{b_n,n}(A_{n,t,\delta})$ converges to $1$ as $n \to \infty$. Next, by Theorem \ref{spec_concentrate},
\begin{equation*}\lim_{n \to \infty} \mu_{b_n,n} \left(\{Q \in \rmg_{b_n}(n) \ | \ \sigma(P_0 Q P_0) \cap K\} \right) = 0,\end{equation*}
where
\begin{equation*} K = [-  \delta, \lambda_- - \delta] \cup [\lambda_+ + \delta, 1 +  \delta],\end{equation*}hence $\lim_{n \to \infty} \mu_{b_n,n}(B_{n, \delta}) = 1$.
\end{proof}
We derive Theorem \ref{grassmann_theorem} as a consequence of a more general result, as follows.

\begin{theorem}\label{maintheorem} Let $f \in \mca$. Then $I^f_{\alpha,\beta} = \lim_{n \to \infty} I^{f}_{n, \alpha, \beta}$ is well defined, and
\begin{enumerate}
\item{If $\alpha = \beta = \frac 1 2$ then $I_{f,\alpha,\beta} = \max_{[0,1]} \psi_f = M_f$,}
\item{If $\alpha < \beta = 1-\alpha$, then $I_{f,\alpha,\beta} = \max\left\{|\alpha_{10}|, \max_{[\lambda_-, 1]} \psi_f\right\}$,}
\item{If $\alpha = \beta < 1-\alpha$, then $I_{f,\alpha,\beta} = \max\left\{|\alpha_{11}|, \max_{[0, \lambda_+]} \psi_f\right\}$,}
\item{If $\alpha < \beta < 1-\alpha$, then $I_{f,\alpha,\beta} = \max\left\{|\alpha_{10}|, |\alpha_{11}|, \max_{[\lambda_-, \lambda_+]} \psi_f\right\}$.}
\end{enumerate}
We recall that $\lambda_- = 0 \Leftrightarrow \alpha = \beta$ and $\lambda_+ = 1 \Leftrightarrow \beta = 1-\alpha$ (and that $\lambda_- < \lambda_+$).
\end{theorem}
\begin{proof}
As before, we use the notation of (\ref{canonic_decomp}), and assume that $\mu_{b_n,n}(G_{n,0}) = 1$.

Fix $\varepsilon > 0$. Then
\begin{equation*} I^{f}_{n,\alpha,\beta} = \int_{\rmg_{b_n}(n)} \Vert f(P_{0},Q) \Vert_{\op}d\mu_{b_n,n}(Q) = \int_{\rmg_{b_n}(n)} \left(\max_{\sigma_{P_{0}, Q}} \Psi_f \right)d\mu_{b_n, n}(Q),\end{equation*}
where $\Psi_f$ and $\sigma_{P_{0}, Q}$ are as in Conclusion \ref{reformulated_norm_formula}. If $\dim V_{10} = 0$ and $\dim V_{11} = 0$, then $\max_{\sigma_{P_{0}, Q}} \Psi_f = \max_{\sigma(I-H_{0})} \psi_f$. The latter observation will guide our proof, i.e., we will first show that
\begin{equation}\label{psi_f_limit} \tilde I^f_{\alpha,\beta} = \lim_{n \to \infty} \int_{\rmg_{b_n}(n)}\left( \max_{\sigma(I-H_{0})} \psi_f \right) d\mu_{b_n,n}(Q) = \max_{[\lambda_-, \lambda_+]} \psi_f.\end{equation}

Let $t_{\max} \in [\lambda_-, \lambda_+]$ be such that
\begin{equation*} M = \max_{[\lambda_-,\lambda_+]} \psi_f = \psi_f(t_{\max}).\end{equation*}
By the continuity of $\psi_f$, there exists $\delta > 0$ such that if $|t - t_{\max}| < \delta$ then
\begin{equation}\label{for_lower_bound} |\psi_f(t) - M| < \frac \varepsilon 2,\end{equation}
 and additionally,

\begin{equation}\label{for_upper_bound}\max_{[\lambda_-- \delta, \lambda_+ + \delta] \cap [0,1]} \psi_f - M < \frac \varepsilon 2.\end{equation}
Note that the latter is trivial if $\lambda_- = 0$, $\lambda_+ = 1$ (since then the left hand side equals $0$). We will use (\ref{for_lower_bound}) to show $\tilde I^f_{\alpha,\beta} \ge M$, and (\ref{for_upper_bound}) to show $\tilde I^f_{\alpha, \beta} \le M$.

Using Conclusion \ref{spec_conclusion}, there exists $n_0\in \NN$ such that for all $n > n_0$,
\begin{equation}\label{also_for_lower_bound} \left(M - \frac \varepsilon 2\right) \mu_{b_n, n}\left(G_{n,0} \cap A_{n, t_{\max}, \delta }\right) > M - \varepsilon.\end{equation}
Thus, since $\max_{\sigma_{P_{0}, Q}}\Psi_f \ge \max_{\sigma(I-H_{0})} \psi_f \ge 0$, we get that for all $n > n_0$,
\begin{gather*} M_f \ge I^{f}_{n,\alpha, \beta} = \int_{\rmg_{b_n}(n)} \left(\max_{\sigma_{P_{0}, Q}} \Psi_f \right) d\mu_{b_n,n} \\ \ge \int_{G_{n,0} \cap A_{n, t_{\max}, \delta }} \left(\max_{\sigma(I-H_{0})} \psi_f\right) d\mu_{b_n,n} \ge M-\varepsilon.\end{gather*}
The last inequality holds first since $\max_{\sigma(I-H_0)} \psi_f > M - \frac \varepsilon 2$ for every projection $Q \in G_{n,0} \cap A_{n, t_{\max}, \delta}$ by (\ref{for_lower_bound}), and then using (\ref{also_for_lower_bound}).

If $\alpha = \beta = \frac 1 2$, then $\lambda_- = 0$, $\lambda_+ = 1$, so $M = M_f$ and $\lim_{n \to \infty} I^{f}_{n,\alpha,\beta} = M_f$ as required. Otherwise, by Conclusion \ref{spec_conclusion}, for $\varepsilon_1 > 0$ which satisfies $M_f \varepsilon_1 < \frac \varepsilon 2$, there exists $n_1 \in \NN$ such that for all $n > n_1$, it holds that
\begin{equation}\label{also_for_upper_bound} \mu_{b_n,n}(C_{n, \delta}) > 1- \varepsilon_1,\end{equation} where
\begin{gather*} C_{n, \delta} =   G_{n,0} \cap A_{n, t_{\max}, \delta}  \cap B_{n, \delta},\\ B_{n,\delta} = \{Q \in \rmg_{b_n}(n)\ | \ \sigma(I-H_{0}) \subset (\lambda_- - \delta,  \lambda_++ \delta)\}.\end{gather*}
Thus,
\begin{gather*} \int_{\rmg_{b_n}(n)} \left(\max_{\sigma(I-H_{0})} \psi_f \right) d\mu_{b_n,n} = \\ \int_{C_{n, \delta}} \left( \max_{\sigma(I-H_{0})} \psi_f\right)  d\mu_{b_n,n} +\int_{G_{n,0} \cap (\rmg_{b_n}(n) \setminus C_{n, \delta})} \left(\max_{\sigma(I-H_{0})} \psi_f\right) d\mu_{b_n, n} \\ \le M + \frac \varepsilon 2 + \varepsilon_1 M_f < M + \varepsilon.\end{gather*}
Here, the first integral is not greater than $M + \frac \varepsilon 2$ since $\max_{\sigma(I-H_0)} \psi_f \le M + \frac \varepsilon 2$ for every $Q \in C_{n,\delta}$ using (\ref{for_upper_bound}), and the second integral is bounded from above by $\varepsilon_1 M_f$ since $\max_{\sigma(I-H_0)} \psi_f \le M_f = \max_{[0,1]} \psi_f$ and then using (\ref{also_for_upper_bound}).

The above completes the proof of (\ref{psi_f_limit}), and the proof for the case $\alpha = \beta = \frac 1 2$.
We turn to address the remaining cases, namely, when $\alpha + \beta = 1,\ \alpha < \beta$, and when $\alpha + \beta < 1,\ \alpha = \beta$ and finally when $\alpha + \beta < 1,\ \alpha < \beta$. Note that for $n$ large enough, if $\alpha < \beta$ then $\dim V_{10} > 0$ for every $Q \in \rmg_{b_n}(n)$, and if $\alpha + \beta < 1$ then $\dim V_{11} > 0$ for every $Q \in \rmg_{b_n}(n)$.

Assume that $\alpha + \beta = 1$ and that $\alpha < \beta$. Let $M_{-} = \max\{|\alpha_{10}|, M\}$. If $M_{-} = |\alpha_{10}|$, then for $n$ large enough, \begin{equation*} \max_{\sigma_{P_{0}, Q}}\Psi_f = |\alpha_{10}|\end{equation*} for every $Q \in \rmg_{b_n}(n)$, hence clearly $I^{f}_{n,\alpha,\beta} = |\alpha_{10}|$ for every $n$ large enough. Otherwise, if $|\alpha_{10}| < M$, we assume without loss of generality that $M - |\alpha_{10}| > \varepsilon$. Then by the above,
\begin{equation*} M - \frac \varepsilon 2 \le \max_{\sigma_{P_{0}, Q}} \Psi_f = \max_{\sigma(I-H_0)} \psi_f \le M + \frac \varepsilon 2\end{equation*}
for every $Q\in C_{n, \delta}$, hence by the above, $\lim_{n \to \infty} I^{f}_{n,\alpha,\beta} = M$.

Similarly, if $\alpha = \beta < \frac 1 2$, then either $M_{+} = \max\{|\alpha_{11}|, M\} = |\alpha_{11}|$, in which case we note that for $n$ large enough, $\max_{\sigma_{P_{0}, Q} \Psi_f} = |\alpha_{11}|$ for every $Q \in \rmg_{b_n}(n)$ so $I^{f}_{n,\alpha,\beta} = |\alpha_{11}|$ for every $n$ large enough, or $M_{+} = M$, in which case as for the previous case, $\lim_{n \to \infty} I^{f}_{n, \alpha, \beta} = M$. 

Finally, if $\alpha + \beta < 1$, $\alpha < \beta$, then either $M_{-,+} = \max\{|\alpha_{10}|, |\alpha_{11}|, M\} = \max\{|\alpha_{10}|, |\alpha_{11}|\}$, in which case we note that for $n$ large enough, $\max_{\sigma_{P_{0}, Q}} \Psi_f = M_{-,+}$ for all $Q \in \rmg_{b_n}(n)$, hence $I^{f}_{n,\alpha,\beta} = M_{-,+}$ for all $n$ sufficiently large, or otherwise $M_{-,+} = M$, in which case we repeat the reasoning above to obtain the required.
\end{proof}
We obtain Theorem \ref{grassmann_theorem} immediately. Note that if $\alpha = \beta \le \frac 1 2$ then (see Theorem \ref{limit_joint} for the formulas of $\lambda_-, \lambda_+$) we obtain
\begin{equation}\label{high_end_point} \lambda_- = 0,\ \lambda_+ = 4\alpha(1-\alpha) = \lambda_\alpha.\end{equation}
\begin{conclusion} If $f \in \ker (T)$, then by Lemma \ref{abelianization}, $|\alpha_{lk}| = 0$ for $l,k \in \{0,1\}$, hence the limits in the distinct cases addressed in Theorem \ref{maintheorem} admit the same formula, namely,
\begin{equation*} \lim_{n \to \infty} I^{f}_{n,\alpha,\beta} = \max_{[\lambda_-, \lambda_+]} \psi_f.\end{equation*} If $\alpha = \beta$, then $[\lambda_-, \lambda_+] = [0,\lambda_\alpha]$, giving Theorem \ref{grassmann_theorem}.\end{conclusion}
\section{Proof of Theorem \ref{spin_theorem}}\label{spin_section}
The present section contains the proof of Theorem \ref{spin_theorem}. In the next subsection, we prove (for the sake of completeness) some basic lemmas about weak and strong convergence in Hilbert spaces. In subsection \ref{spin_preprint_results}, we specify the limits of "central" matrix coefficients of the projections, as computed in a previous work (\cite{shabtai}) by the author.  The contents of subsections \ref{preliminaries}, \ref{spin_preprint_results} suffice to reduce (in subsections \ref{l2Zsubsection}, \ref{l2Tsubsection}) the evaluation of the limit to the case of two projections on $L^2(\TT)$ (where $\TT\subset \CC$ is the unit circle), namely, the Cauchy-Szeg\"{o} projection on the Hardy space (denoted $\Pi_\TT$), and the operator of multiplication by the indicator function of $E_\alpha = \{0<\Re \zeta < 2 \alpha\} \subset \TT$ (denoted $\mcm_{\mathbbm 1_{E_\alpha}}$). Finally, in subsection \ref{cauchyindicator}, we use a classical result on the spectrum of Toeplitz operators with bounded symbols (together with Theorem \ref{normformula}) to conclude the proof.
\subsection{Preliminaries}\label{preliminaries}
Let $\mcb(\mch)$ denote the space of bounded operators on a complex Hilbert space $\mch$. We will use the following elementary notions and facts.
\begin{definition} Let $\{A_n\}_{n \in \NN} \subset \mcb(\mch)$ and $A \in \mcb(\mch)$.
\begin{enumerate}
\item{We say that $\{A_n\}_{n\in  \NN}$ converges strongly to $A$ if $\lim_{n \to \infty} A_n v = A v$ for every $v \in \mch$.}
\item{We say that $\{A_n\}_{n \in \NN}$ converges weakly to $A$ if $\lim_{n \to \infty} \langle A_n u, v \rangle = \langle A u, v \rangle$ for every $u, v \in \mch$.}
\end{enumerate}
\end{definition}
\begin{lemma}\label{basistoWOT} Let $\mce = \{e_k \ | \ k \in \NN\}$ denote an orthonormal basis of $\mch$. Let $\{A_n \}_{n \in \NN} \subset \mcb(\mch)$ and $A \in \mcb(\mch)$. Then $\{A_n \}_{n \in \NN}$ converges weakly to $A$ if and only if $\sup_n \Vert A_n \Vert_{\op} < \infty$ and $\lim_{n \to \infty} \langle A_n e_l, e_k \rangle = \langle A e_l, e_k \rangle$ for every $l,k \in \NN$.\end{lemma}
\begin{proof}
It is well known that weakly convergent sequences are bounded. Indeed, for $v \in \mch$ let $\psi_{n,v}(u) = \langle A_n v, u \rangle$, $\psi_{n,v} \in \mch^*$. Then $\{\lvert \psi_{n,v}(u) \rvert\}_{n \in \NN}$ is bounded for every $u \in \mch$, hence $\{\Vert \psi_{n,v} \Vert_{\mch^*}\}_{n \in \NN} = \{\Vert A_n v \Vert_{\mch}\}_{n \in \NN}$ is bounded by the uniform boundedness principle. Since $\{\Vert A_n v \Vert\}_{n \in  \NN}$ is bounded for all $v \in \mch$, then again using uniform boundedness, we deduce that $\{\Vert A_n \Vert_{\op}\}_{n \in \NN}$ is bounded.

Conversely, assume that $\Vert A_n \Vert_{\op} < M$ for all $n$, let $u,v \in \mch$ and let $\varepsilon > 0$. For $\delta > 0$, we may take \begin{equation*} \tilde u = \sum_{k \le k_0} \langle u, e_k \rangle e_k,\ \tilde v = \sum_{k \le k_0} \langle v, e_k \rangle e_k\end{equation*} such that $\Vert u - \tilde u \Vert < \delta,\ \Vert v - \tilde v \Vert < \delta$. Then, for sufficiently large $n$ so that $\lvert \langle (A_n - A) \tilde u, \tilde v \rangle \rvert < \delta$, we have that
\begin{gather*} \lvert \langle A_n u, v \rangle - \langle A u, v \rangle \rvert =\\ \lvert (A_n - A)(u - \tilde u), v \rangle + \langle(A_n - A) \tilde u, v - \tilde v \rangle + \langle (A_n - A) \tilde u, \tilde v \rangle \rvert \\ \le \Vert A_n - A \Vert_{\op} \left(\Vert v \Vert + \Vert \tilde u \Vert \right) \delta + \delta \le (M + \Vert A \Vert_{\op})(\Vert v \Vert + \Vert u \Vert ) \delta + \delta <  \varepsilon,\end{gather*}
where the latter holds for $\delta$ sufficiently small.
\end{proof}
\begin{lemma}\label{weaktostrong} Let $\mce = \{e_k \ | \ k \in \NN\}$ denote an orthonormal basis of $\mch$. If $\{A_n\}_{n \in \NN} \subset \mcb(\mch)$ converges weakly to $A \in \mcb(\mch)$ and $\lim_{n \to \infty} \Vert A_n e_k \Vert = \Vert A e_k \Vert$ for every (fixed) $k \in \NN$, then $A_n$ converges strongly to $A$. \end{lemma}
\begin{proof}
We note that
\begin{equation*} \Vert A_n e_k - A e_k \Vert^2 = \Vert A_n e_k \Vert^2 + \Vert A e_k \Vert^2 - \langle A e_k, A_n e_k \rangle - \langle A_n e_k, A e_k \rangle,\end{equation*}
hence by weak convergence,
\begin{equation*} \lim_{n \to \infty} \Vert A_n e_k - A e_k \Vert^2 = 0,\end{equation*}
that is, $\lim_{n \to \infty} A_n e_k = A e_k$ for every fixed $k \in \NN$.

Let $v = \sum_{k \in \NN} v_k e_k \in \mch$. Let $\varepsilon > 0$. Note that $\sup_n \Vert A_n \Vert_{\op} < \infty$. Hence there exists $k_0 > 0$ such that $v_\varepsilon = \sum_{k \ge k_0} v_k e_k$ satisfies $\Vert (A - A_n) v_\varepsilon \Vert < \frac \varepsilon 2$. Let $u_\varepsilon = v - v_\varepsilon$. Then
\begin{equation*} \lim_{n \to \infty} A_n u_\varepsilon = A u_\varepsilon,\end{equation*}
so there exists $n_0$ such that for all $n \ge n_0$, it holds that $\Vert (A_n - A)u_\varepsilon \Vert < \frac \varepsilon 2$. Thus for all $n \ge n_0$,
\begin{equation*} 0 \le \Vert (A_n - A) v \Vert = \Vert(A_n - A)(u_\varepsilon + v_\varepsilon) \Vert < \varepsilon,\end{equation*}
as required.
\end{proof}
\begin{lemma} Assume that $\{A_n\}_{n \in \NN} \subset \mcb(\mch)$ converges to $A \in \mcb(\mch)$ strongly. Then $\liminf \Vert A_n \Vert_{\op} \ge \Vert A \Vert_{\op}$ (this is also true if $A_n$ converges to $A$ weakly). \end{lemma}
\begin{proof}
Let $\varepsilon > 0$. Assume that $v \in \mch$ satisfies $\Vert A v \Vert > \Vert A \Vert_{\op} - \frac 1 2 \varepsilon$. Then there exists $n_0$ such that for all $n > n_0$, it holds that $\Vert A_n v \Vert > \Vert A \Vert_{\op} - \varepsilon$. Hence the required.
\end{proof}
\begin{conclusion}\label{normsconverge} If $\{A_n\}_{n \in \NN} \subset \mcb(\mch)$ converges to $A \in \mcb(\mch)$ strongly and $\Vert A_n \Vert_{\op} \le \Vert A \Vert_{\op}$ for all $n$ then $\lim_{n \to \infty} \Vert A_n \Vert_{\op} = \Vert A \Vert_{\op}$. \end{conclusion}
\begin{lemma}\label{SOTproduct} Assume that $\{A_n\}_{n \in  \NN}$ and $\{B_n \}_{n \in \NN}$ converge to $A, B$ strongly. Then $A_n B_n$ converges strongly to $A B$. \end{lemma}
\begin{proof}
Let $M_A > 0$ be such that $\Vert A_n \Vert_{\op} < M_A$ for all $d$.
Let $v \in \mch$. Then
\begin{gather*} \Vert (A_n B_n - A B) v \Vert = \Vert (A_n B_n - A_n B + A_n B - A B) v \Vert \\ \le \Vert A_n (B_n-B) v \Vert + \Vert (A_n - A) B v \Vert\\ \le M_A \Vert (B_n - B) v \Vert + \Vert(A_n - A) B v \Vert,\end{gather*}
hence by strong convergence of $A_n, B_n$,
$\lim_{n \to \infty} \Vert (A_n B_n - A B) v \Vert = 0$. \end{proof}
\subsection{Matrices of spectral projections of spin operators}\label{spin_preprint_results}
We consider the standard basis $\mcb_{3,n} = \left\{e_j, e_{j-1}, ..., e_{-j}\right\}$ of eigenvectors of $J_3$, where as before $2j+1 = n$. Recall that the spectrum of $J_1, J_2, J_3$ is
\begin{equation*} \sigma_n = \{j, j-1, ..., -j\},\end{equation*}
and denote $P_{1,\alpha,n} = \mathbbm 1_{(0,\alpha_n)}(J_1)$ and $P_{3,\alpha,n} = \mathbbm 1_{(0,\alpha_n)}(J_3)$, where $\alpha_n > 0$ satisfies
\begin{equation*} \#\left( (0, \alpha_n) \cap \sigma_n\right) = \lfloor \alpha n \rfloor\end{equation*}
for every $n \in \NN$. Let $[A]_{\mcb_{3,n}}$ denote the matrix of $A : \CC^{n} \to \CC^{n}$ relative to $\mcb_{3,n}$. Clearly,
\begin{equation}\label{P3matrix} \left[P_{3,\alpha,n}\right]_{\mcb_{3,n}} = \left(\begin{array}{cc} \tilde I_{\alpha, n} & 0 \\ 0 & 0 \end{array}\right),\end{equation}
where
\begin{equation*} \tilde I_{\alpha, n} = \left(\begin{array}{cc} 0 & 0 \\ 0 & I_{\lfloor \alpha n \rfloor} \end{array}\right).\end{equation*}
Write $\left[P_{1, \alpha, n}\right]_{\mcb_{3,n}} = \left(P_{1,\alpha, n, m', m} \right)_{m', m = j, j-1, ..., -j}$.
\begin{theorem}[\cite{shabtai}]\label{entries_limit} Fix $m', m \in \sigma_n$. Let $\widehat{\mathbbm 1}_{E_\alpha}(k)$ denote the $k$-th Fourier coefficient of the indicator function of $E_{\alpha} = \{0 < \Re \zeta < 2 \alpha \} \subset \TT$. Then
\begin{equation*} \lim_{k \to \infty}  P_{1,\alpha,n+2k, m',m} = \widehat{\mathbbm 1}_{E_\alpha}(m-m').\end{equation*} \end{theorem}
In particular,
\begin{conclusion}\label{basis_norm_limit} Fix $m \in \sigma_n$. Then $\lim_{k \to \infty} \left\Vert P_{1,\alpha,n+2k} e_m \right\Vert^2 = \widehat{\mathbbm 1}_{E_{\alpha}}(0)$. \end{conclusion}
\begin{proof}
Note that
\begin{gather*} \langle P_{1,\alpha,n} e_m, P_{1,\alpha,n} e_m \rangle = \langle P_{1,\alpha,n}^* P_{1,\alpha,n} e_m, e_m \rangle \\ = \langle P_{1,\alpha,n} e_m, e_m \rangle = P_{1,\alpha,n,m,m}.\end{gather*}
Thus, by Theorem \ref{entries_limit}, we obtain the required.
\end{proof}
\subsection{The $l^2(\ZZ)$ settings}\label{l2Zsubsection}
It will be convenient to work in $L^2(\TT)$ rather than $\CC^{n}$. As an intermediate step, we consider $l^2(\ZZ)$. Let $\widehat \mcb = \left\{\widehat e_k \ | \ k \in \ZZ\right\}$ denote the standard basis of $l^2(\ZZ)$. Define an embedding $\Psi_n : \CC^n \to l^2(\ZZ)$ by
\begin{equation*} \Psi_n\left(e_m \right) = \left\{\begin{array}{ll} \widehat e_{m - \frac 1 2} & n \in 2\NN \\ \widehat e_{m-1} & n \in 2\NN-1 \end{array}\right.,\ m = j, j-1, ..., -j,\ n = 2j+1. \end{equation*}
Let $V_{n} = \Psi_n\left(\CC^n\right)$, and let $\Pi_n : l^2(\ZZ) \to V_{n} \subset l^2(\ZZ)$ denote the orthogonal projection on $V_{n}$.
For an operator $A : \CC^{n} \to \CC^{n}$ we set
\begin{equation*} A^{\ZZ} = \Psi_n \circ A \circ \Psi^{-1}_n \circ \Pi_n.\end{equation*}
\begin{conclusion}\label{matrixelements} If the matrix of $A$ in $\mcb_{3,n}$ is $\left(a_{m',m}\right)_{|m|, |m'| \le j}$, then the matrix elements of $A^{\ZZ}$ in the basis $\widehat \mcb$ are as follows.

If $n \in 2\NN-1$, then
\begin{equation*} \langle A^{\ZZ} \widehat e_l, \widehat e_k \rangle = \left\{\begin{array}{ll} 0 & |l+1| > j \text{ or } |k+1| > j,\\  a_{k+1,l+1} & |l+1|, |k+1| \le j\end{array}\right. \end{equation*}

If $n \in 2\NN$, then
\begin{equation*} \langle A^{\ZZ} \widehat e_l, \widehat e_k \rangle = \left\{\begin{array}{ll} 0 & \left|l+ \frac 1 2\right | > j \text{ or } \left| k + \frac 1 2 \right| > j \\ a_{k + \frac 1 2, l + \frac 1 2} & \left| l + \frac 1 2 \right|, \left| k + \frac 1 2 \right| \le j \end{array}\right. \end{equation*}
\end{conclusion}
\subsection{The $L^2(\TT)$ settings}\label{l2Tsubsection}
Let $\mcb = \{\zeta^k \ | \ k \in \ZZ\}$ denote the standard orthonormal basis of $L^2(\TT)$, which we identify with $l^2(\ZZ)$ in the obvious way. Let $P^{\TT}_{1,\alpha,n}, P^{\TT}_{3,\alpha,n}$ be the equivalents of $P^{\ZZ}_{1,\alpha,n}, P^{\ZZ}_{3,\alpha,n}$. Let $\Pi_\TT : L^2(\TT) \to H^2(\TT)$ be the orthogonal Cauchy-Szeg\"{o} projection on the Hardy space $H^2(\TT) \subset L^2(\TT)$. Finally, for $\psi \in L^\infty(\TT)$ let $\mcm_\psi : L^2(\TT) \to L^2(\TT)$ be the multiplication operator $G \mapsto \psi G$.
\begin{claim} $P^{\TT}_{3,\alpha,n}$, $P^{\TT}_{1,\alpha,n}$ converge strongly to $\Pi_\TT$, $\mcm_{\mathbbm 1_{E_{\alpha}}}$ respectively. \end{claim}
\begin{proof}
Clearly $\Vert P^{\TT}_{3,\alpha,n} \Vert_{\op} = \Vert P^{\TT}_{1,\alpha,n} \Vert_{\op} = 1$ for all $n \in \NN$. By (\ref{P3matrix}), Claim \ref{entries_limit} and Conclusion \ref{matrixelements}, we have for every $k,l\in \ZZ$
\begin{gather*} \lim_{n \to \infty} \langle P^{\TT}_{3,\alpha,n} \zeta^l, \zeta^k \rangle_{L^2(\TT)} = \langle \Pi_{\TT} \zeta^l, \zeta^k \rangle_{L^2(\TT)},\\ \lim_{n \to \infty} \langle P^{\TT}_{1, \alpha, n} \zeta^l, \zeta^k \rangle_{L^2(\TT)} = \langle \mcm_{\mathbbm 1_{E_{\alpha}}}, \zeta^l, \zeta^k \rangle_{L^2(\TT)}.\end{gather*}
Thus, Lemma \ref{basistoWOT} gives us weak convergence. Then, as evident in (\ref{P3matrix}),
\begin{equation*} \lim_{n \to \infty} P^{\TT}_{3, \alpha, n} \zeta^k = \Pi_{\TT} \zeta^k\end{equation*} for all $k \in \ZZ$. Also,
\begin{equation*} \left \Vert \mcm_{\mathbbm 1_{E_{\alpha}}} \zeta^k \right \Vert^2 = \langle \mathbbm{1}_{E_{\alpha}} \zeta^k, \mathbbm 1_{E_{\alpha}} \zeta^k\rangle_{L^2(\TT)} = \langle \mathbbm 1_{E_{\alpha}}, 1 \rangle_{L^2(\TT)} = \widehat {\mathbbm 1}_{E_{\alpha}}(0).\end{equation*}
Thus, noting Conclusion \ref{basis_norm_limit}, we see that Lemma \ref{weaktostrong} applies to both $\{P_{1,\alpha,n}^{\TT}\}_{n \in \NN}$ and $\{P_{3,\alpha,n}^{\TT}\}_{n \in \NN}$. \end{proof}
\begin{conclusion} $f(P^{\TT}_{3,\alpha,n}, P^{\TT}_{1,\alpha,n})$ converges strongly to $f(\Pi_{\TT}, \mcm_{\mathbbm 1_{E_\alpha}})$. \end{conclusion}
\begin{proof}
This follows for all monomials by induction using Lemma \ref{SOTproduct}, then for all polynomials using linearity.
\end{proof}
\subsection{The pair $\Pi_\TT, \mcm_{\mathbbm 1_{E_{\alpha}}}$}\label{cauchyindicator}
We study $P =\Pi_{\TT}$, $Q= \mcm_{\mathbbm 1_{E_{\alpha}}}$ in the context of the general theory of pairs of orthogonal projections. In the notations of Theorem \ref{canonic_repr},
\begin{equation*} P Q P = (1,0,0,0) \oplus \left(\begin{array}{cc} I - H & 0 \\ 0 & 0 \end{array}\right).\end{equation*}
\begin{definition} The Toeplitz operator associated with the symbol $\phi \in L^\infty(\TT)$ is $T_\phi = \Pi_\TT \mcm_\phi \Pi_\TT : H^2(\TT) \to H^2(\TT)$.\end{definition}
Thus,
\begin{equation*} P Q P = (1,0,0,0) \oplus (I-H) \oplus 0 = T_{\mathbbm 1_{E_{\alpha}}}.\end{equation*}
We note the following classical result.
\begin{theorem}[Hartman-Wintner, \cite{hartmanwintner}, \cite{douglas}, 7.20] If $\phi \in L^\infty(\TT)$ is real-valued, then the spectrum of $T_\phi$ is given by
\begin{equation*} \sigma(T_\phi) = [\ess \inf \phi, \ess \sup \phi].\end{equation*} \end{theorem}
\begin{conclusion}\label{final_conc} Hartman-Wintner's Theorem implies that
\begin{equation*} \sigma(I-H) = \sigma(H) = [0,1].\end{equation*}
Thus, by Theorem \ref{normformula} and Conclusion \ref{universal_upper_bound}, $\Vert f(P,Q) \Vert_{\op} = M_f$. The polynomial $f(P^{\TT}_{3,\alpha,n}, P^{\TT}_{1,\alpha,n})$ converges to $f(P,Q)$ strongly, hence by Conclusion \ref{normsconverge}
\begin{equation*}\lim_{n \to \infty} \Vert f(P^{\TT}_{3,\alpha,n}, P^{\TT}_{1,\alpha,n}) \Vert_{\op} = M_f,\end{equation*} 
and since $\Vert f(P_{3,\alpha,n}, P_{1,\alpha,n}) \Vert_{\op} = \Vert f(P^{\TT}_{3,\alpha,n}, P^{\TT}_{1,\alpha,n}) \Vert_{\op}$, this completes the proof of Theorem \ref{spin_theorem}.\end{conclusion}
\begin{remark} By the F. and M. Riesz Theorem, if $G$ is holomorphic on the unit disk $\mathbb D \subset \CC$, with zero radial boundary values on a subset $E \subset \TT$ of positive Lebesgue measure, then $G \equiv 0$. Thus, the same holds for an anti-holomorphic function. Hence (in the notation of (\ref{canonic_decomp})), for the projections $P,Q$, we have $V_{lk} = \{0\}$ for every $l,k \in \{0,1\}$, so that in fact
\begin{equation*} PQP = T_{\mathbbm 1_{E_{\alpha}}} = I - H : H^2(\TT) \to H^2(\TT).\end{equation*}\end{remark}
\section{Proof of Theorem \ref{prolate_theorem}}\label{prolate_section}
Denote $R_{n} = P_{3,\alpha,n} P_{1,n} P_{3,\alpha, n} \in \End\left(\text{Im}(P_{3,\alpha,n}) \right)$. Let $\lambda_1 \ge \lambda_2 \ge ... \ge \lambda_{\lfloor \alpha n \rfloor}$ be the eigenvalues of $R_{n}$. We will prove, by direct computation, that
\begin{equation*} \trace\left(R_n\right) - \trace\left(R_n^2\right)= \sum_{k=1}^{\lfloor \alpha n \rfloor} \lambda_k(1-\lambda_k) = \Theta(\log n),\end{equation*}which readily implies Theorem \ref{prolate_theorem}. This method is quite standard in the context of Slepian spectral concentration problems; specifically, our proof follows \cite{EMT}.

We denote the matrix of $P_{1,n}$ with respect to the eigenbasis $\mcb_{3,n}$ of $J_3$ by
\begin{equation*} \left[ P_{1,n} \right]_{\mcb_{3,n}} = \left(p^j_{m',m}\right)_{m',m = j, j-1, ..., -j},\ n = 2j+1.\end{equation*}

\begin{claim}\label{offdiagonal} Fix $0 < \gamma < 1$. Let $S_n$ be a sub-matrix of $\left[P_{1,n} \right]_{\mcb_{3,n}}$ of the form
\begin{equation*} S_n = \left(p^j_{m',m}\right)_{\substack{m_0 > m' \\ m \ge m_0}},\end{equation*}
where $|m_0| \le \gamma j$. Then the Frobenius norm of $S_n$ satisfies $\Vert S_n \Vert_F^2 = \Theta(\log n)$ (with constants depending only on $\gamma$).\end{claim}
The claim will be established through a series of lemmas, and our main tool is the formula (\cite{shabtai})
\begin{equation}\label{px_entries} p^j_{m',m} =\left\{\begin{array}{ll}\frac 1 2 e^{i \frac \pi 2(m'-m)} \left(\delta_{m',m} - \widehat{d^j_{m',m}}(0)\right) & \text{if } m'-m \in 2\ZZ,\\ -\frac i 2 e^{i \frac \pi 2 (m'-m)} \mch_\TT (d^j_{m',m})(0) & \text{if } m'-m \in 2\ZZ + 1 \end{array}\right.,\end{equation}
where $\HH_\TT : L^2(\TT) \to L^2(\TT)$ is the periodic Hilbert transform and $d^j_{m',m}: \TT \to \TT$ are the Wigner small-d functions, specified by
\begin{equation*} d^j_{m',m}(\theta) = \langle e^{-i \theta J_2} e_m, e_{m'} \rangle.\end{equation*} The zeroth Fourier coefficient of $d^j_{m',m}$ is specified by (\cite{shabtai}, \cite{biedenharnlouck}, 3.78, \cite{fengwangyangjin})
\begin{equation}\label{zerothfourier} \widehat{d^j_{m',m}}(0) = \left\{\begin{array}{ll} e^{i\frac \pi 2 (m-m')} d^j_{m,0} \left(\frac \pi 2 \right) d^j_{m',0} \left(\frac \pi 2 \right) & \text{if } j \in \NN,\\ 0 & \text{if } j\in \frac 1 2 \NN \setminus \NN\end{array}\right..\end{equation}
Moreover, when $j \in \NN$ (\cite{varshalovich}, 4.16(6)),
\begin{equation}\label{for_zerothfourier_formula} d^j_{m,0} \left(\frac \pi 2 \right) =\left\{\begin{array}{ll} \frac{(-1)^{\frac{j+m} 2}}{2^j} \sqrt{{j+m \choose \frac{j+m} 2}{j-m \choose \frac{j-m} 2}} & \text{if } j-m \in 2\NN \\ 0 & \text{if } j-m \in 2\NN + 1 \end{array}\right..\end{equation} Throughout, we will make use of the symmetry relations (\cite{varshalovich}, 4.4) 
\begin{equation}\label{dj_symm_1} d^j_{m',m}(-\theta) = (-1)^{m'-m} d^j_{m',m}(\theta) = d^j_{m,m'}(\theta) = d^j_{-m',-m}(\theta),\end{equation}
and
\begin{equation}\label{dj_symm_2} d^j_{m',m}(\theta + 2\pi) = (-1)^{j-m} d^j_{m', -m}(\theta+\pi) = (-1)^{2j} d^j_{m',m}(\theta).\end{equation}
\begin{lemma}\label{difference_even} $\sum_{m'-m \in 2\ZZ \setminus \{0\}} (p^j_{m',m})^2 = \mco(1)$. \end{lemma}
\begin{proof}
Combining (\ref{px_entries}) and (\ref{zerothfourier}), we see that if $j \not \in \NN$ then $p^j_{m',m} = 0$ for every $m,m'$ such that $m'-m \in 2\ZZ \setminus \{0\}$. Otherwise (assuming $m' \ne m$),
\begin{equation*} p^j_{m',m} =-\frac 1 2 \widehat{d^j_{m',m}}(0)= -\frac 1 2 d^j_{m',0} \left(\frac \pi 2 \right) d^j_{m,0} \left( \frac \pi 2 \right).\end{equation*}
If we fix $m'\in \{j, j-1, ..., -j\}$, then
\begin{equation*} \sum_{m'-m \in 2\ZZ \setminus \{0\}} (p^j_{m',m})^2 \le \frac 1 4 \left(d^j_{m',0} \left( \frac \pi 2 \right)\right)^2 \sum_{m = -j}^{j} \left(d^j_{m,0} \left( \frac \pi 2 \right)\right)^2 = \frac 1 4 d^j_{m',0} \left(\frac \pi 2 \right)^2,\end{equation*}
where we used that $d^j_{m',m}\left(\frac \pi 2 \right)\in \RR$ are the elements of a unitary matrix, which also implies that
\begin{equation*} \frac 1 4 \sum_{m'=-j}^j \left(d^j_{m',0} \left( \frac \pi 2 \right)\right)^2 = \frac 1 4,\end{equation*}
hence the required.
\end{proof}
Next, by the symmetries (\ref{dj_symm_1}), we may assume that $m_0 \ge 1$ without loss of generality.
\begin{claim} When $j \in \NN$, $\Vert S_n \Vert_F^2= \Theta(\log n)$. \end{claim}
\begin{proof}
%
The norms of the rows and columns of $S_n$ are obviously bounded from above by $1$, hence we can erase $\mco(1)$ of them without loss of generality. Thus, we will estimate the Frobenius norm of
\begin{equation*} \tilde S_n = \left(p^j_{m',m} \right)_{\substack{m_0> m' >-j+2\\ j-2 > m \ge m_0}},\end{equation*}
and furthermore assume that $m_0 \in 2\NN +1$.

The previous lemma implies that it suffices to consider the entries of $\tilde S_n$ corresponding to $m'-m \in 2\ZZ + 1$.
Since $j\in \NN$, the functions $d^j_{m',m}$ are $2\pi$-periodic (\ref{dj_symm_2}), so that
\begin{equation*} \mch_\TT(d^j_{m',m})(0) =-\frac 1 {2\pi} \int_{-\pi}^{\pi} d^j_{m',m}(\theta) \cot \left(\frac \theta 2 \right) d\theta.\end{equation*}
Denote
\begin{equation*} \nu_{j,m} = \sqrt{(j+m)(j-m+1)} = \nu_{j,-(m-1)}.\end{equation*}
By \cite{biedenharnlouck}, 3.85,
\begin{equation*} \cot \left(\frac \theta 2 \right) d^j_{m',m}(\theta) = \frac 1 {m-m'} \left(\nu_{j,m} d^j_{m', m-1}(\theta)+ \nu_{j,m'} d^j_{m'-1, m}(\theta)\right)\end{equation*}
Using (\ref{zerothfourier}),
we find that
\begin{gather*} p^j_{m',m} = \\ \frac 1 {2(m-m')} \left( \nu_{j,m} d^j_{m',0} \left(\frac \pi 2 \right) d^j_{m-1, 0} \left(\frac \pi 2 \right) - \nu_{j,m'} d^j_{m'-1, 0} \left( \frac \pi 2 \right) d^j_{m,0} \left(\frac \pi 2 \right) \right).\end{gather*}

Now, (\ref{for_zerothfourier_formula}) implies, in particular, that either $d^j_{m,0}\left(\frac \pi 2 \right) = 0$ or $d^j_{m-1, 0} \left(\frac \pi 2 \right) = 0$.
Hence, we assume from now on that  $j=2l\in 2\NN$ (the case $j\in 2\NN + 1$ is essentially identical), so that
\begin{equation*} p^j_{m',m} = \left\{\begin{array}{cc} -\frac {\nu_{j,m'}} {2(m-m')} d^j_{m'-1,0} \left(\frac \pi 2 \right) d^j_{m,0} \left( \frac \pi 2 \right) & m' \in 2\ZZ + 1\\ \frac {\nu_{j,m}}{2(m-m')} d^j_{m',0} \left( \frac \pi 2 \right) d^j_{m-1, 0} \left( \frac \pi 2 \right) & m' \in 2\ZZ \end{array}\right. .\end{equation*}
Write $m_0 = 2r+1 \in 2\NN + 1$. Then we must estimate
\begin{equation*} \sum_{\substack{1 \le m' \le l+r-1\\ 1 \le m \le l-r-2}} (p^{2l}_{2r-2m'+1, 2r+2m})^2 + \sum_{\substack{1 \le m' \le l+r-1\\ 1 \le m \le l-r-1}} (p^{2l}_{2r-2m'+2, 2r + 2m-1})^2.\end{equation*}
We next show that the first sum is $\Theta(\log l)$. The same can be shown, in the same way, for the second sum, hence we omit the computation.
We use (\ref{for_zerothfourier_formula})
together with the estimate
\begin{equation}\label{central_binomial} {2k \choose k} =\frac 1 {\sqrt{\pi k}} 4^k \left( 1 + \frac{\mco(1)} k \right)\end{equation}
to obtain
\begin{gather*} (p^{2l}_{2r-2m'+1, 2r + 2m})^2 =  \frac{\Theta(1)}{(2m+2m'-1)^2} \frac {(2(l + r - m')+1)(l-(r-m')} {\sqrt{(l^2 -(r-m')^2)(l^2 -(r+m)^2)}}\\ =  \frac {\Theta(1)}{(2m+2m'-1)^2} \frac{\sqrt{l^2-(r-m')^2}}{\sqrt{l^2-(r+m)^2}} \le \frac{\Theta(1)}{(2m+2m'-1)^2} \frac l {\sqrt{l^2 - (r+m)^2}}.\end{gather*}
Write $t = r+m$, $s = m'-r$. Then
\begin{equation*} \sum_{\substack{1 \le m' \le l+r-1\\ 1 \le m \le l-r-2}}(p^{2l}_{2r-2m'+1, 2r+2m})^2 = \Theta(1)\sum_{\substack{1-r \le s \le l-1 \\ 1+r \le t \le l-2}}\frac 1{(2t+2s-1)^2} \frac{\sqrt{l^2 - s^2}}{\sqrt{l^2 - t^2}}.\end{equation*}
Let $\varepsilon = \frac 1 2 (1-\gamma)$. If $s \ge 1-r + \varepsilon l$, then $s + t \ge \varepsilon l$, so that
\begin{gather*} \sum_{\substack{s \ge 1-r+\varepsilon l \\ 1+r \le t \le l-2}} \frac 1{(2t + 2s -1)^2} \frac{\sqrt{l^2 - s^2}}{\sqrt{l^2 - t^2}} \le \mco(1) \sum_{\substack{s \ge 1-r + \varepsilon l \\ 1+r \le t \le l-2}}\frac 1 l \frac 1 {\sqrt{l^2 - t^2}} \\ \le \mco(1) \sum_{t = 1}^{l-1} \frac 1 {\sqrt{l^2 - t^2}} = \mco(1).\end{gather*}
Next, assume $s < 1-r + \varepsilon l$. Note that
\begin{equation*}  m_0 = 2r + 1 < (1-2\varepsilon) 2l,\end{equation*}
hence $ r \le (1-2\varepsilon) l-1$, so that $1+r + \varepsilon l \le (1-\varepsilon) l$.
If
\begin{equation*} 1 + r + \varepsilon l \le t \le l-1,\end{equation*}
then
\begin{equation*} t + s \ge t + 1 - r > \varepsilon l,\end{equation*}
and the same argument as above implies that
\begin{equation*}  \sum_{\substack{1-r \le s < 1-r + \varepsilon l \\ t \ge 1+r + \varepsilon l}}\frac 1 {(2t + 2s - 1)^2} \frac {\sqrt{l^2 - s^2}}{\sqrt{l^2 - t^2}} = \mco(1).\end{equation*}
Finally, assume $t < 1 + r + \varepsilon l \le (1-\varepsilon) l$. Then
$l^2 - t^2 = \Theta(l^2)$,
so that
\begin{equation*}\frac{\sqrt{l^2 - s^2}}{\sqrt{l^2 - t^2}} = \Theta(1).\end{equation*}
Thus
\begin{gather*}  \sum_{\substack{1-r \le s < 1-r+\varepsilon l \\ 1+r \le t < 1+r + \varepsilon l}} \frac 1 {(2t + 2s - 1)^2} \frac{\sqrt{l^2 - s^2}}{\sqrt{l^2 - t^2}} = \\ \Theta(1) \sum_{\substack{1-r \le s < 1-r + \varepsilon l \\ 1+r \le t < 1+r + \varepsilon l}} \frac 1 {(2t + 2s - 1)^2}= \Theta(\log l) = \Theta(\log n),\end{gather*}
as required. Namely, the original sum is $\Theta(\log n)$.\end{proof}
\begin{conclusion} Assume $j \in \frac 1 2 \NN \setminus \NN$. Then $\Vert S_n \Vert_F^2 = \mco(\log n)$. \end{conclusion}
\begin{proof}
If $m'-m \in 2\ZZ$, then $p^j_{m',m} = 0$ by (\ref{px_entries}) and (\ref{zerothfourier}) (we know that $m' \ne m$). Hence, assume that $m'-m \in 2\ZZ + 1$. The functions $d^j_{m',m}$ satisfy
\begin{equation*} d^j_{m',m}(\theta + 2\pi) = - d^j_{m',m}(\theta),\end{equation*}
hence
\begin{gather*} \mch_\TT(d^j_{m',m})(0) = -\frac 1 {4\pi} \int_{-2\pi}^{2\pi} d^j_{m',m}(\theta) \cot\left(\frac \theta 4 \right) d\theta \\ = -\frac 1 {2\pi} \int_0^{2\pi} d^j_{m',m}(\theta)\frac 1 {\sin\left(\frac \theta 2 \right)} d\theta.\end{gather*}
As before, we can truncate $\mco(1)$ rows and columns from $S_n$ without changing our estimate. According to \cite{biedenharnlouck}, 3.83,
\begin{gather*} \frac 1 {\sin\left(\frac \theta 2 \right)} d^j_{m',m} =\sqrt{\frac{j+m}{j-m'}} d^{j-\frac 1 2}_{m' + \frac 1 2, m - \frac 1 2}(\theta) + \sqrt{\frac{j-m}{j-m'}} \cot \left(\frac \theta 2 \right) d^{j - \frac 1 2}_{m' + \frac 1 2, m+ \frac 1 2}(\theta),\end{gather*}
where $\sqrt{\frac{j\pm m}{j-m'}} = \mco(1)$ (since $m' < m_0 \le \gamma j$), and again using \cite{biedenharnlouck}, 3.85,
\begin{gather*} \cot \left(\frac \theta 2 \right) d^{j - \frac 1 2}_{m' + \frac 1 2, m + \frac 1 2}(\theta) = \\ \frac 1 {m-m'} \left(\nu_{j-\frac 1 2, m} d^{j-\frac 1 2}_{m' + \frac 1 2, m - \frac 1 2}(\theta) + \nu_{j-\frac 1 2, m'} d^{j-\frac 1 2}_{m' - \frac 1 2, m + \frac 1 2}(\theta)\right).\end{gather*}
Thus,
\begin{gather*} \frac 1 {\sin \left( \frac \theta 2 \right)} d^j_{m',m}(\theta) =\left(\sqrt{\frac{j+m}{j-m'}} + \sqrt{\frac{j-m}{j-m'}} \frac{\nu_{j-\frac 1 2, m}}{m-m'} \right)d^{j- \frac 1 2}_{m' + \frac 1 2, m - \frac 1 2}(\theta) \\ + \sqrt{\frac{j-m}{j-m'}} \frac{\nu_{j-\frac 1 2, m'}}{m-m'} d^{j- \frac 1 2}_{m' - \frac 1 2, m + \frac 1 2}.\end{gather*}
If $m > \frac{1+ \gamma} 2 j$, then $m-m' = \Theta(j)$. Thus,
\begin{equation*} \frac 1 {\sin \left(\frac \theta 2 \right)} d^j_{m',m}(\theta) = \mco(1) \left(d^{j- \frac 1 2}_{m'+ \frac 1 2, m-\frac 1 2}(\theta) + d^{j-\frac 1 2}_{m'- \frac 1 2, m + \frac 1 2}(\theta) \right),\end{equation*} which implies that
\begin{equation*} p^j_{m',m} = \mco(1) \left(\widehat{d^{j-\frac 1 2}_{m' + \frac 1 2, m-  \frac 1 2}}(0) + \widehat{d^{j-\frac 1 2}_{m' - \frac 1 2, m + \frac 1 2}}(0) \right),\end{equation*} where at most one of the summands is non-zero by (\ref{for_zerothfourier_formula}), so
\begin{equation*} (p^j_{m',m})^2 = \mco(1) \left(\left(\widehat{d^{j-\frac 1 2}_{m' + \frac 1 2, m-  \frac 1 2}}(0)\right)^2 + \left(\widehat{d^{j-\frac 1 2}_{m' - \frac 1 2, m + \frac 1 2}}(0)\right)^2 \right).\end{equation*}
However, as in Lemma \ref{difference_even},
\begin{equation*} \sum_{m',m} \left(\widehat{d^{j-\frac 1 2}_{m'-\frac 1 2, m+\frac 1 2}}(0)\right)^2 = 1,\end{equation*}
where the sum is over all applicable $m',m$. Hence, we deduce that
\begin{equation*} \sum_{\substack{m_0 > m'\\ m > \frac{1 + \gamma} 2 j}} (p^j_{m',m})^2 = \mco(1).\end{equation*}If $m \le \frac{1 + \gamma} 2 j$, then
\begin{equation*} \sqrt{\frac{j-m}{j-m'}} = \Theta(1),\ \nu_{j-\frac 1 2, m} = \sqrt{(j-m)(j+m)} = \Theta(j),\end{equation*}
so
\begin{gather*} \frac 1 {\sin \left(\frac \theta 2 \right)} d^j_{m',m}(\theta) =\\ \Theta(1) \left(  \left(1 + \frac{\nu_{j- \frac 1 2, m}}{m-m'} \right)d^{j- \frac 1 2}_{m' + \frac 1 2, m - \frac 1 2}(\theta) + \frac{\nu_{j-\frac 1 2, m'}}{m-m'}d^{j- \frac 1 2}_{m'- \frac 1 2, m + \frac 1 2}(\theta)  \right)\\ = \frac{\Theta(1)}{m-m'}\left(\nu_{j- \frac 1 2, m} d^{j- \frac 1 2}_{m' + \frac 1 2, m - \frac 1 2}(\theta) + \nu_{j-\frac 1 2, m'}d^{j- \frac 1 2}_{m'- \frac 1 2, m + \frac 1 2}(\theta)\right) \\= \Theta(1)\cot \left(\frac \theta 2 \right) d^{j- \frac 1 2}_{m' + \frac 1 2, m + \frac 1 2},\end{gather*}
hence
\begin{equation*} p^j_{m',m} =\Theta(1) p^{j-\frac 1 2}_{m'+\frac 1 2, m+\frac 1 2},\end{equation*} and so the squared Frobenius norm of the sub-matrix of $S_n$ corresponding to $m \le \frac{1 + \gamma} 2j$ is $\Theta(\log j)$ from the previous claim. 
\end{proof}
%
%
\begin{conclusion} Let $P_{n,a,b} = \mathbbm 1_{(aj, bj)}(J_3)$, where either $a > -1$ or $b <1$. For $s < t$, let $N_n(s,t)$ denote the number of eigenvalues of
\begin{equation*} R_{n,a,b} = P_{n,a,b} P_{1,n} P_{n,a,b}\in \End\left(\text{Im}(P_{n,a,b})\right)\end{equation*} lying in the interval $[s,t]$. Then for every $0 < t < \frac 1 2$
\begin{equation*} \lim_{n \to \infty} \frac 1 {(b-a) n} N_n(0,t) = \lim_{n \to \infty} \frac 1 {(b-a) n} N_n(1-t, 1) = \frac 1 4,\end{equation*}
and
\begin{equation*} N_n(t, 1-t) = \mco(\log n).\end{equation*}  \end{conclusion}
\begin{proof}
Let $\lambda_1 \ge \lambda_2 \ge ... \ge \lambda_{k_{n,a,b}}$ denote the eigenvalues of $R_{n,a,b}$.
First of all, we note that
\begin{equation*} \sum_{k = 1}^{k_{n,a,b}} \lambda_k = \trace R_{n,a,b} = \sum_{aj < m < bj} p^j_{m,m},\end{equation*}
where by (\ref{px_entries}), (\ref{zerothfourier}) and (\ref{for_zerothfourier_formula}),
\begin{equation*} p^j_{m,m} = \frac 1 2 - \widehat{d^j_{m,m}}(0) = \left\{\begin{array}{ll} \frac 1 2 - \frac 1 2 \left(d^j_{m,0}\left(\frac \pi 2 \right) \right)^2 & j \in \NN \\ \frac 1 2 & j \not \in \NN \end{array}\right..\end{equation*}
Thus when $j \not \in \NN$,
\begin{equation}\label{sum_eig_vals}\trace R_{n,a,b} = \frac{b-a} 4 n + \mco(1).\end{equation}
The same holds for $j \in \NN$, since $\sum_{m = -j}^j \left(d^j_{m,0}\left(\frac \pi 2 \right) \right)^2 = 1$.
Next, we wish to estimate
\begin{equation*} \sum_{k = 1}^{k_{n,a,b}} \lambda_k^2 = \trace\left(R_{n,a,b}^2 \right) = \Vert R_{n,a,b} \Vert_F^2.\end{equation*}
Instead of estimating the Frobenius norm of $R_{n,a,b}$, we can estimate that of its "complement"
\begin{equation*} S_{n,a,b} = (I - P_{n,a,b}) P_{1,n} P_{n,a,b}.\end{equation*}
Indeed, denote $v_m = \left[P_{1,n} e_m \right]_{\mcb_{3,n}}$, where $m = j, j-1, ..., -j$. Then
\begin{equation*} \Vert v_m \Vert^2 = \langle P_{1,n} e_m, P_{1,n} e_m \rangle = p^j_{m,m},\end{equation*}
so that
\begin{equation*} \Vert S_{n,a,b} \Vert_F^2 + \Vert R_{n,a,b} \Vert_F^2 = \sum_{aj < m < bj} \Vert v_m \Vert^2 = \trace(R_{n,a,b}).\end{equation*}
Now, write $\text{Im}(P_{n,a,b}) = \Sp\{e_m \ | \ m_+ \ge m \ge m_-\}$, and let $P_+$, $P_-$ denote the orthogonal projections on $\Sp\{e_m \ | \ m > m_+\}$, $\Sp\{e_m \ | \ m < m_-\}$ respectively. Then $S_{n,a,b} = S_{n,a,b}^+ + S_{n,a,b}^-$, where
\begin{equation*} S^+_{n,a,b} = P_+ P_{1,n} P_{n,a,b},\ S_{n,a,b}^-= P_- P_{1,n} P_{n,a,b}\end{equation*}
have "off-diagonal" matrices as addressed in Claim \ref{offdiagonal}. Thus,
\begin{equation*} \Vert S_{n,a,b} \Vert_F^2 = \Vert S_{n,a,b}^+ \Vert_F^2 + \Vert S_{n,a,b}^- \Vert_F^2 = \Theta(\log n).\end{equation*}
We conclude that
\begin{equation*} \sum_{k =1}^{k_{n,a,b}} \lambda_k(1-\lambda_k) = \trace(R_{n,a,b}) - \trace\left(R_{n,a,b}^2 \right) = \Theta(\log n).\end{equation*}
Now, if $0< t < \lambda < 1-s $, then $\lambda(1-\lambda) > t s$, hence
\begin{equation*} \Theta(\log n) = \sum_{k=1}^{k_{n,a,b}} \lambda_k(1-\lambda_k) \ge N_n(t,1-s) t s,\end{equation*}
so
\begin{equation}\label{log_eig_vals}N_n(t,1-s) = \mco(\log n).\end{equation} Next, write
\begin{equation*} x_{n,t} = \frac 1 {(b-a) n} N_n(0,t),\ z_{n,t} = \frac 1 {(b-a) n} N_n(1-t, t).\end{equation*}If $t > s$ then $x_{n,t} \ge x_{n,s}$ and $z_{n,t} \ge z_{n,s}$.  We know that
\begin{equation*}N_n(0,t) + N_n(t,1-t) + N_n(1-t,1) =k_{n,a,b} = \frac{b-a} 2 n + \mco(1),\end{equation*} hence using (\ref{log_eig_vals}),
\begin{equation}\label{xntznt}  x_{n,t} + z_{n,t} = \frac 1 2 + o(1).\end{equation}
Given $\varepsilon >0$, choose $0 < \delta < t$ such that
\begin{equation*}  \frac{\delta^2}{1-\delta^2} <  \frac \delta {1-\delta} < \varepsilon.\end{equation*}
Let
\begin{equation*} \sigma^-_{n,t} = \{k \ | \ \lambda_k \le t\},\ \sigma^+_{n,t} = \{k \ | \ \lambda_k \ge 1-t\}.\end{equation*} Then by (\ref{sum_eig_vals}) and (\ref{log_eig_vals}),
\begin{gather*} \frac 1 4 + o(1) = \frac 1 {(b-a) n} \left(\sum_{k \in \sigma^-_{n,\delta}} \lambda_k + \sum_{k \in \sigma^+_{n,\delta}} \lambda_k \right) \ge (1-\delta) z_{n,\delta},\end{gather*}
so
\begin{equation*} z_{n,\delta} \le \frac 1 {4(1-\delta)} + o(1) = \frac 1 4 + \frac{\delta}{4(1-\delta)} + o(1) < \frac 1 4 +\frac \varepsilon 4 + o(1).\end{equation*}
Thus,
\begin{equation*} z_{n,t} \le z_{n,\delta} + \frac 1 {(b-a)n} N_n(1-t, 1-\delta) < \frac 1 4 + \varepsilon,\end{equation*}
where the last inequality is true for every $n$ sufficiently large. Similarly,
\begin{equation*} \frac 1{(b-a)n} \left(\sum_{k \in \sigma^-_{n,\delta}} \lambda_k^2 + \sum_{k \in \sigma^+_{n,\delta}} \lambda_k^2 \right)\le \delta^2 x_{n,\delta} + z_{n,\delta} = \frac {\delta^2} 2 + z_{n,\delta}(1-\delta^2) + o(1),\end{equation*}
using (\ref{xntznt}), hence
\begin{equation*} z_{n,t} \ge z_{n,\delta} \ge \frac{1-2\delta^2}{4(1-\delta^2)}+o(1) = \frac 1 4 - \frac 1 4 \frac{\delta^2}{1-\delta^2} + o(1) > \frac 1 4 - \varepsilon,\end{equation*} where the last inequality holds for every $n$ large enough. Thus $\lim_{n \to \infty}z_{n,t} = \frac 1 4$. Finally, $x_{n,t} = \frac 1 2 - z_{n,t} + o(1) = \frac 1 4 + o(1)$ as well.
\end{proof}
\section{Concluding remarks}\label{discussion_section}
\subsection{Analogues of Theorem \ref{spin_theorem}}
The proof of Theorem \ref{spin_theorem} relies on the fact that in the semiclassical limit $n \to \infty$, the spectral projections $P_{1,\alpha,n}$ and $P_{3,\alpha,n}$ converge to $\mcm_{\mathbbm 1_{E_\alpha}}$ and $\Pi_\TT$ in some appropriate sense. Analogous facts hold (in particular) for the rest of the pairs of spectral projections addressed in \cite{shabtai}, hence the proof can be adapted so as to apply to them as well.

Notably, we considered pairs of spectral projections coming from position and momentum operators
\begin{equation*} \widehat q = \mcm_q,\ \widehat p = - i \hbar \frac \partial {\partial q}\end{equation*}
on $L^2(\RR)$. We also considered pairs of spectral projections corresponding to the operators $\cos \widehat \theta = \mcm_{\cos \theta}$ and $\cos \widehat L$, where

\begin{equation*} \widehat \theta = \mcm_\theta,\ \widehat L = -i\frac {2\pi} n \frac{\partial}{\partial \theta}\end{equation*}
are the analogues (\cite{mukunda, cylinder1}) of $\widehat q$, $\widehat p$ on $L^2(\TT) \simeq L^2 \left([0, 2\pi], \frac {d\theta}{2\pi} \right)$.  Yet another example involved the generators $g_1, g_2$ of the finite Heisenberg groups $H(\ZZ_n)$\footnote{Here, $\ZZ_n = \ZZ / n\ZZ$.} (\cite{schwinger, vourdas, varadaraweisbart}), which act on $F\in l^2(\ZZ_n) \approx L^2(\TT)$ by
\begin{equation*} g_1 F(k) = e^{\frac{2\pi k} n i} F(k),\ g_2 F(k) = F(k+1).\end{equation*}

Let $E = \{\zeta \in \TT \ | \ \Re \zeta > 0\}$. Then the pair of projections
\begin{equation*} \mathbbm 1_{(0,\infty)}(\widehat q) = \mcm_{\mathbbm 1_{(0,\infty)}}, \mathbbm 1_{(0,\infty)}(\widehat p)\end{equation*}is unitarily equivalent to $\mcm_{\mathbbm 1_E}$, $\Pi_\TT$ independently of $\hbar$ (since $\mathbbm 1_{(0,\infty)}(\widehat p)$ is just the projection on the Hardy space $H^2(\RR)$). For the pair \begin{equation*}\mathbbm 1_{(0,\infty)}(\cos \widehat \theta) = \mcm_{\mathbbm 1_E},\ \mathbbm 1_{(0,\infty)}(\cos \widehat L),\end{equation*} a sequence of unitary operators $U_n : L^2(\TT) \to L^2(\TT)$ is required in order to reduce the proof to the case of $\mcm_{\mathbbm 1_E}$ and $\Pi_\TT$. Similarly, for \begin{equation}\label{heis_projs}\Pi_1 = \mathbbm 1_{(0,\infty)}(\Re g_1),\ \Pi_2 = \mathbbm 1_{(0,\infty)}(\Re g_2),\end{equation} where $\Re A = \frac 1 2(A + A^*)$, we need a sequence of embeddings $\tilde U_n : l^2(\ZZ_n) \to L^2(\TT)$. $U_n$, $\tilde U_n$ are defined by mapping elements of the standard bases of $L^2(\TT)$, $l^2(\ZZ_n)$ to those of $L^2(\TT)$ in a certain suitable way, so as to obtain the convergence of the relevant spectral projections to the pair $\Pi_\TT$, $\mcm_{\mathbbm 1_E}$. The arguments are essentially the same as those of subsections \ref{l2Zsubsection}, \ref{l2Tsubsection}. Finally, applying Conclusion \ref{final_conc} to $\Pi_\TT$, $\mcm_{\mathbbm 1_E}$, we obtain the following.

\begin{conclusion}\label{analogues_conclusion} Analogues of Theorem \ref{spin_theorem} hold for the families of pairs of spectral projections detailed above, i.e.,
\begin{gather*} \lim_{n \to \infty} \left \Vert f(\Pi_1,\Pi_2) \right \Vert_{\op} = \lim_{n \to \infty} \left \Vert f\left(\mcm_{\mathbbm 1_E}, \mathbbm 1_{(0,\infty)}(\cos \widehat L)\right) \right \Vert_{\op} =\\ \left \Vert f\left(\mcm_{\mathbbm 1_{(0,\infty)}}, \mathbbm 1_{(0,\infty)}\left(\widehat p\right)\right) \right \Vert_{\op} = M_f\end{gather*}
for every $f \in \mca$.\end{conclusion}
\subsection{Angles between subspaces and convergence rate}\label{angles_conv_subsection}
Let $\mch$ be a complex, finite-dimensional Hilbert space. Assume that $P : \mch \to V_P$ and $Q : \mch \to V_Q$ are two orthogonal projections with $\dim V_P \le \dim V_Q$. Let us for now consider $PQP$ as an element of $\End(V_P)$. Then the canonical form (Theorem \ref{canonic_repr}) implies that
\begin{equation*} PQP = \Id_{V_{00}} \oplus 0_{V_{01}} \oplus \left(\Id - H \right).\end{equation*}
Thus (see subsection \ref{angles_subsection}), the eigenvalues $\lambda_1 \ge ... \ge \lambda_{\dim V_P}$ of $PQP$ are given by $\lambda_k = \cos^2 \phi_k$, where $\phi_1 \le ... \le \phi_{\dim V_P}$ are the principal angles between the subspaces $V_P$, $V_Q$.
\begin{conclusion}\label{norm_angles}If $f \in \ker(T)$, then using Conclusion \ref{reformulated_norm_formula},
\begin{equation*} \Vert f(P,Q) \Vert_{\op} = \max_{\sigma(PQP)} \psi_f = \max_{\phi \in \Phi} \psi_f(\cos^2 \phi),\end{equation*}
where $\Phi$ is the set of principal angles between $V_P, V_Q$.\end{conclusion}
Theorem \ref{prolate_theorem} and Conclusion \ref{random_angles} imply that there is a striking discrepancy between the spectrum of random $P Q P$ and that of $P_{3,n} P_{1,n} P_{3,n}$ (as illustrated in figures \ref{fig7} and \ref{fig8}).
According to \cite{EMT}, the same is true for the spectral projections $\Pi_1$, $\Pi_2$ specified in (\ref{heis_projs}). These discrepancies imply that the rates of convergence in Theorem \ref{spin_theorem} (for $P_{3,n}$, $P_{1,n}$) and in Conclusion \ref{analogues_conclusion} (for $\Pi_1$, $\Pi_2$) tend to be much slower than in Theorem \ref{grassmann_theorem} (for random $P,Q$).

Indeed, consider the case of commutators ($f(x,y) = xy - yx$), depicted in figures \ref{fig1}, \ref{fig2}. Then $\psi_f(t) = \sqrt{t(1-t)}$ by Example \ref{example_psi_f}, hence by Conclusion \ref{norm_angles}, \begin{equation*} \Vert [P,Q] \Vert_{\op} = \max_{\phi \in \Phi} \sin \phi \cos \phi = \sin \phi_{0} \cos \phi_{0},
\end{equation*} where $\phi_{0}\in \Phi$ is the principal angle closest to $\frac \pi 4$.

\begin{remark}\label{conv_rate_conc}
Let $\Phi_n$ be the set of principal angles between $\text{Im}\left(P_{1,n}\right)$, $\text{Im}\left(P_{3,n} \right)$. While Theorem \ref{spin_theorem} implies that $\Phi_n$ becomes dense in $\left[0, \frac \pi 2\right]$ as $n \to \infty$, this happens very slowly as per Theorem \ref{prolate_theorem} (logarithmically, or slower). Thus, $\min_{\phi \in \Phi_n} \lvert \phi - \frac \pi 4 \rvert$ tends to $0$ slowly, hence $\Vert \left[P_{1,n}, P_{3,n} \right] \Vert_{\op}$ tends to $\frac 1 2$ slowly.

By contrast, the set of principal angles $\Phi$ associated with a "typical" random pair $P,Q\in \rmg_{\lfloor \frac n 2 \rfloor}(n)$ is distributed much more evenly (as per Conclusion \ref{random_angles}). Thus, $\min_{\phi \in \Phi} \lvert \phi - \phi_0 \rvert$ is typically quite small for every $\phi_0\in \left[0, \frac \pi 2 \right]$ and in particular for $\phi_0 = \frac \pi 4$, hence $\int_{\Omega_n}\left \Vert [P,Q ] \right \Vert_{\op}d\nu_n$ converges quickly to $\frac 1 2$.

Similarly for general $f \in \ker(T)$.\end{remark}
The last remark is relevant for any sequence of pairs of projections which behaves as in theorems \ref{spin_theorem}, \ref{prolate_theorem} (e.g., $P_{3,\alpha,n}$, $P_{1,n}$, and likely also $P_{3,\alpha,n}$, $P_{1,\alpha,n}$).
\subsection{A conjecture}
Let us offer an informal, conjectured explanation for "maximality results" of the type of Theorem \ref{spin_theorem} and Conclusion \ref{analogues_conclusion}. The explanation is based on the notion of quantization, and is inspired by findings from \cite{pol1, charpol}. In what follows, $\mcl(\mch)$ denotes the space of self-adjoint operators on a Hilbert space $\mch$.

Let $(M, \omega)$ denote a closed\footnote{i.e., compact and without boundary.}, quantizable\footnote{i.e., $\frac \omega {2\pi}$ represents an integral de-Rham cohomology class.} symplectic manifold. A Berezin-Toeplitz quantization (\cite{charles, lefloch, schlichen}) of $M$ produces a sequence of finite dimensional complex Hilbert spaces $\mch_n$, such that $\lim_{n \to \infty} \dim \mch_n = + \infty$, together with surjective linear maps $T_n : C^\infty(M) \to \mcl(\mch_n)$. The maps are required to satisfy several desirable properties in the semiclassical limit $n \to \infty$.
\begin{example} Up to normalization, $J_1 = T_n(x_1)$ and $J_3 = T_n(x_3)$, where $x_1, x_3 : S^2 \to \RR$ are the Cartesian coordinate functions. \end{example}
Let $F_1, F_2 \in C^\infty(M)$ and assume that $I_1, I_2 \subset \RR$ are a pair of non-trivial intervals. We consider the spectral projections
\begin{equation*} \Pi_{1,n} = \mathbbm 1_{I_1}\left(T_n(F_1) \right),\ \Pi_{2, n} = \mathbbm 1_{I_2}\left(T_n(F_2) \right)\end{equation*}
as a pair of quantum observables that are somehow related (\cite{zz2}) to the domains
\begin{equation*} D_1 = F_1^{-1}(I_1),\ D_2 = F_2^{-1}(I_2).\end{equation*}
\begin{example} In this interpretation, $P_{1,\alpha,n}$, $P_{3,\alpha,n}$ are "associated" with the domains $\{0 < x_1 \le 2\alpha\}, \{0 < x_3 \le 2\alpha\} \subset S^2$. \end{example}
Recall that $M_f$ denotes the universal, tight upper bound for $\Vert f(P,Q) \Vert_{\op}$, where $P,Q$ are arbitrary orthogonal projections on a separable complex Hilbert space. Our various numerical simulations appear to support the following.
\begin{conjecture} Fix $f \in \ker(T)$. Assume that $F_1, F_2$ are Poisson non-commuting and that $M$ is two-dimensional. 
\begin{enumerate}
\item{If $\partial D_1 \cap \partial D_2 \ne \emptyset$ is transversal, then
\begin{equation*} \lim_{n \to \infty} \left \Vert f\left(\Pi_{1,n}, \Pi_{2,n} \right) \right \Vert_{\op} = M_f.\end{equation*}}
\item{If the distance between $\partial D_1, \partial D_2$ is greater than some $\varepsilon > 0$, then
\begin{equation*} \lim_{n \to \infty} \left \Vert f\left(\Pi_{1,n}, \Pi_{2,n} \right) \right \Vert_{\op} = 0.\end{equation*}}
\end{enumerate}
\end{conjecture}
The latter is essentially a conjecture about the principal angles between subspaces spanned by eigenstates of quantum observables, or equivalently, about the spectrum of $\Pi_{1,n} \Pi_{2,n} \Pi_{1,n}$ (see Conclusion \ref{norm_angles}). When $\partial D_1 \cap \partial D_2 \ne \emptyset$ is transversal, we expect the eigenvalues of $\Pi_{1,n} \Pi_{2,n} \Pi_{1,n}$ to cluster near $0$, $1$, similarly to the situation specified in Theorem \ref{prolate_theorem}. If true, this would constitute an interesting formulation of the Slepian spectral concentration phenomenon.

We refer the reader to \cite{shabtai} (the final section in particular) for further details and simulations (mostly involving spin operators, but also some simulations for finite Heisenberg groups).
\renewcommand{\abstractname}{Acknowledgements}
\begin{abstract}
This research has been partially supported by the European Research Council Starting Grant 757585 and by the Israel Science Foundation grant 1102/20. I wish to express my sincere gratitude to the European Research Council and to the Israel Science Foundation.\hfill

I wish to thank my advisors Leonid Polterovich and Lev Buhovski for many meetings and discussions, valuable comments, and general guidance in this project. I wish to thank Boaz Klartag, Mikhail Sodin and Sasha Sodin for significant discussions and comments, and Dor Elboim for many useful discussions and for his extremely important assistance in tackling several difficulties along the way.

Finally, I wish to thank David Kazhdan for his suggestion to consider the topics addressed in this work, and for providing the conjecture which initiated this project. The findings presented here would not have been possible without his involvement and insight.\end{abstract}
\textsc{School of Mathematical Sciences, Tel Aviv University, Ramat Aviv, Tel Aviv 6997801 Israel}

E-mail address: oodshabt@post.tau.ac.il

\begin{thebibliography}{9}
%
\bibitem{biedenharnlouck}
L.C. Biedenharn, J.D. Louck,
\newblock{\em Angular Momentum in Quantum Physics: Theory and Application},
\newblock Encyclopedia of Mathematics and its Applications, {\bf 8}, Addison-Wesley Publishing Company, Reading, MA, (1981).

\bibitem{bottspitk}
A. B\"{o}ttcher, I.M. Spitkovsky,
\newblock{\em A gentle guide to the basics of two projections theory},
\newblock Linear Algebra Appl. {\bf 432} (2010), 1412-1459.
%
\bibitem{charles}
L. Charles,
\newblock{\em Quantization of Compact Symplectic Manifolds},
\newblock J. Geom. Anal. {\bf 26} (2016), 2664-2710.
%
\bibitem{charpol}
L. Charles, L. Polterovich,
\newblock{\em Sharp correspondence principle and quantum measurements}.
\newblock Algebra i Analiz {\bf 29} (2017), no. 1, 237-278
%
\bibitem{collins}
B. Collins,
\newblock{\em Product of random projections, Jacobi ensembles and universality problems arising from free probability},
\newblock Probab. Theory Related Fields {\bf 133} (2005), 315-344.
%
\bibitem{douglas}
R. G. Douglas,
\newblock{\em Banach Algebra Techniques in Operator Theory},
\newblock Graduate Texts in Mathematics, Springer-Verlag, NY, 2nd edition, 1998
%
\bibitem{dumitriupaquette}
I. Dumitriu, E. Paquette,
\newblock{\em Global fluctuations for linear statistics of $\beta$-Jacobi ensembles},
\newblock Random Matrices Theory Appl. {\bf 1} (2012), 1250013
%
\bibitem{EMT}
A. Edelman, P. McCorquodale, S. Toledo,
\newblock{\em The future fast Fourier transform?}
\newblock SIAM J. Sci. Comput. {\bf 20}, no. 3 (1998), 1094-1114

\bibitem{fengwangyangjin}
X.M. Feng, P. Wang, W. Yang, G.R. Jin,
\newblock{\em High-precision evaluation of Wigner's d-matrix by exact diagonalization}.
\newblock Phys. Rev. E {\bf 92} (2015).
%

\bibitem{gileskummer}
R. Giles, H. Kummer,
\newblock{\em A matrix representation of a pair of projections in a Hilbert space},
\newblock Canad. Math. Bull {\bf 14(1)} (1971), 35-44 
%
\bibitem{halmos}
P. Halmos,
\newblock{\em Two subspaces},
\newblock Trans. Amer. Math. Soc. {\bf 144} (1969) 381-389
%
\bibitem{hartmanwintner}
P.Hartman, A. Wintner,
\newblock{\em The spectra of Toeplitz's matrices},
\newblock Amer. J. Math. {\bf 76} (1954), 867-882
%
%
\bibitem{lefloch}
Y. Le Floch,
\newblock{\em A Brief Introduction to Berezin-Toeplitz Operators on Compact K\"{a}hler Manifolds},
\newblock CRM Short Courses, Springer International Publishing, 2014.
%
\bibitem{mukunda}
N. Mukunda,
\newblock{\em Wigner distribution for angle coordinates in quantum mechanics}.
\newblock American Journal of Physics {\bf 47}, 182 (1979).
%

\bibitem{pol1}
L. Polterovich,
\newblock{\em Symplectic geometry of quantum noise}.
\newblock Commun. Math. Phys. {\bf 327} (2014), 481-519.

\bibitem{cylinder1}
M.A. Przanowski, J. Tosiek,
\newblock{\em Remarks on Deformation Quantization On The Cylinder}.
\newblock Acta Physica Polonica B {\bf 31} (2000) 561-587.
%

\bibitem{schlichen}
M. Schlichenmaier,
\newblock{\em Berezin-Toeplitz quantization for compact K\"{a}hler manifolds. A review of results},
\newblock Adv. Math. Phys. (2010), Article ID 927280, doi:10.1155/2010/927280.
%
\bibitem{schwinger}
J. Schwinger,
\newblock{\em Unitary operator bases}.
\newblock Proc. Nat. Acad. Sci. USA {\bf 46} (1960), 570-579.
%
\bibitem{shabtai}
O. Shabtai,
\newblock{\em Commutators of spectral projections of spin operators},
\newblock arXiv:2008.00221

\bibitem{slepian}
D. Slepian, H.O. Pollak,
\newblock{\em Prolate spheroidal wave functions, Fourier analysis and uncertainty V: the discrete case},
\newblock Bell System Tech. J. {\bf 57} (1978), 1371-1430

\bibitem{spitkovsky}
I.M. Spitkovsky,
\newblock{\em Once more on algebras generated by two projections},
\newblock Linear Algebra Appl. {\bf 208/209} (1994), 377-395
%

\bibitem{varadaraweisbart} 
V.S. Varadarajan, D. Weisbard,
\newblock{\em Convergence of quantum systems on grids}.
\newblock J. Math. Anal. Appl. {\bf 336} (2007), 608-624.
%
\bibitem{varah}
J. M. Varah,
\newblock{\em The prolate matrix},
\newblock Linear Algebra Appl. {\bf 187}(1993), 269-278

\bibitem{varshalovich}
D.A. Varshalovich, A.N. Moskalev, V.K. Khernoskii,
\newblock{\em Quantum Theory of Angular Momentum}
\newblock World Scientific, Singapore (1988).
%
\bibitem{vourdas}
A. Vourdas,
\newblock{\em Quantum systems with finite Hilbert space}.
\newblock Rep. Prog. Phys. {\bf 67} (2004), 267-320.

\bibitem{llwang}
L.L. Wang,
\newblock{\em A review of prolate spheroidal wave functions from the perspective of spectral methods},
\newblock J. Math. Study {\bf 50(2)} (2017), 101-143.

\bibitem{zz2}
S. Zelditch, P. Zhou,
\newblock{\em Central Limit Theorem for Spectral Partial Bergman Kernels},
\newblock Geom. Topol., {\bf 23.4} (2019), 1961–2004.
\end{thebibliography}
\end{document}